\pdfoutput=1

\documentclass[
showpacs,
showkeys,
superscriptaddress,
groupedaddress,
longbibliography,
aps,
pra,
twocolumn
]{revtex4-2}

\usepackage[T1]{fontenc}
\usepackage[utf8]{inputenc}
\usepackage[english]{babel}

\usepackage{lipsum}
\usepackage{comment}

\usepackage{graphicx}
\usepackage{epsfig}
\usepackage[all,cmtip,matrix,frame,arrow,pdf]{xy}
\usepackage{tabularx}
\usepackage{booktabs}

\usepackage{amsfonts}
\usepackage{amsthm}
\usepackage{amsmath}
\usepackage{amssymb}
\usepackage{mathrsfs}

\usepackage{soul}

\usepackage{enumerate}

\usepackage{tikz}
\input{tikzit.sty}

\tikzstyle{rectangle}=[fill=white, draw=black, shape=rectangle]
\tikzstyle{red}=[fill=red, draw=black, shape=circle]
\tikzstyle{system_label}=[fill=white, draw=white, shape=circle]

\tikzstyle{lambda}=[-, draw={rgb,255: red,191; green,191; blue,191}]
\tikzstyle{state}=[<-]
\tikzstyle{new edge style 0}=[-]

\usepackage{stmaryrd}
\usepackage{xpatch}

\usepackage[]{hyperref}
\hypersetup{colorlinks, citecolor=blue, urlcolor=blue, linkcolor=blue,linktocpage,breaklinks}


\theoremstyle{definition}
\newtheorem*{definition*}{Definition}
\newtheorem{definition}{Definition}

\newcommand{\textdef}[1]{\textit{{#1}}}

\theoremstyle{plain}

\newtheorem*{property*}{Property}

\newtheorem{postulate}{Postulate}

\newtheorem{theorem}{Theorem}
\newtheorem*{theorem*}{Theorem}

\newtheorem*{corollary*}{Corollary}

\newtheorem*{proposition*}{Proposition}

\newtheorem*{conjecture*}{Conjecture}
\newtheorem*{question*}{Question}

\newtheorem*{problem*}{Problem}

\newtheorem*{lemma*}{Lemma}
\newtheorem{lemma}{Lemma}

\newtheorem*{example*}{Example}

\newtheorem*{remark*}{Remark}

\newtheorem{remark}{Remark}
\newtheorem*{proof*}{Proof}
\newcommand{\proofRight}{\ensuremath{\mathbf{\Longrightarrow ) \; }}} 
\newcommand{\proofLeft}{\ensuremath{\mathbf{\Longleftarrow ) \; }}}

\DeclareSymbolFont{sfletters}{OML}{cmbrm}{m}{it}
\DeclareMathSymbol{\salpha}{\mathord}{sfletters}{"0B}
\DeclareMathSymbol{\sbeta}{\mathord}{sfletters}{"0C}
\DeclareMathSymbol{\sgamma}{\mathord}{sfletters}{"0D}
\DeclareMathSymbol{\sdelta}{\mathord}{sfletters}{"0E}
\DeclareMathSymbol{\sepsilon}{\mathord}{sfletters}{"0F}
\DeclareMathSymbol{\szeta}{\mathord}{sfletters}{"10}
\DeclareMathSymbol{\seta}{\mathord}{sfletters}{"11}
\DeclareMathSymbol{\stheta}{\mathord}{sfletters}{"12}
\DeclareMathSymbol{\siota}{\mathord}{sfletters}{"13}
\DeclareMathSymbol{\skappa}{\mathord}{sfletters}{"14}
\DeclareMathSymbol{\slambda}{\mathord}{sfletters}{"15}
\DeclareMathSymbol{\smu}{\mathord}{sfletters}{"16}
\DeclareMathSymbol{\snu}{\mathord}{sfletters}{"17}
\DeclareMathSymbol{\sxi}{\mathord}{sfletters}{"18}
\DeclareMathSymbol{\spi}{\mathord}{sfletters}{"19}
\DeclareMathSymbol{\srho}{\mathord}{sfletters}{"1A}
\DeclareMathSymbol{\ssigma}{\mathord}{sfletters}{"1B}
\DeclareMathSymbol{\stau}{\mathord}{sfletters}{"1C}
\DeclareMathSymbol{\supsilon}{\mathord}{sfletters}{"1D}
\DeclareMathSymbol{\sphi}{\mathord}{sfletters}{"1E}
\DeclareMathSymbol{\schi}{\mathord}{sfletters}{"1F}
\DeclareMathSymbol{\spsi}{\mathord}{sfletters}{"20}
\DeclareMathSymbol{\somega}{\mathord}{sfletters}{"21}
\DeclareMathSymbol{\svarepsilon}{\mathord}{sfletters}{"22}
\DeclareMathSymbol{\svartheta}{\mathord}{sfletters}{"23}
\DeclareMathSymbol{\svarpi}{\mathord}{sfletters}{"24}
\DeclareMathSymbol{\svarrho}{\mathord}{sfletters}{"25}
\DeclareMathSymbol{\svarsigma}{\mathord}{sfletters}{"26}
\DeclareMathSymbol{\svarphi}{\mathord}{sfletters}{"27}

\def\<{\langle}\def\>{\rangle}

\newcommand{\mathDef}{:=}

\newcommand{\prim}[1]{{#1^{\prime}}}
\newcommand{\secondE}[1]{#1^{\prime \prime}}
\newcommand{\kronekerDelta}[2]{\delta_{ #1 , #2}}

\newcommand{\cartesianC}{\times}
\newcommand{\cartesianProduct}[2]{\ensuremath{#1 \cartesianC #2}}

\newcommand{\inverse}[1]{{#1}^{-1}}

\newcommand{\normGeneric}[2]{\ensuremath{\left\lVert {#1} \right\rVert_{#2}}}
\newcommand{\normOp}[1]{\normGeneric{#1}{op}}

\newcommand{\OPT}{\ensuremath{\Theta}}
\newcommand{\MCT}{\ensuremath{\Theta}}

\newcommand{\rbra}[1]{\left({#1}\right\vert}
\newcommand{\rket}[1]{\left\vert{#1}\right)}

\newcommand{\rbraSystem}[2]{{\left({#1}\right\vert}_{\system{#2}}}
\newcommand{\rketSystem}[2]{{\left\vert{#1}\right)}_{\system{#2}}}
\newcommand{\rbraketSystem}[3]{{\left( #1 \vphantom{#2} \right|\left. #2 \vphantom{#1} \right)}_{\system{#3}}}

\newcommand{\uniDetEff}{e}

\newcommand{\system}[1]{\ensuremath{\mathrm{#1}}}
\newcommand{\systemSequence}[2]{\ensuremath{\mathrm{#1}_{#2}}}
\newcommand{\trivialSystem}{\ensuremath{\system{I}}}
\newcommand{\s}[1]{\ustick{\scriptstyle{\system{#1}}}}
\newcommand{\sSequence}[2]{\ustick{\scriptstyle{\systemSequence{#1}{#2}}}}
\newcommand{\sSequencePrime}[2]{\ustick{\scriptstyle{\systemSequence{#1}{#2}'}}}
\newcommand{\sEnsemble}[3]{\ustick{\scriptstyle{ \left\{ \system{#1}_{#2} \right\}_{#2 \in #3} } }}
\newcommand{\sEnsembleDouble}[5]{\ustick{\scriptstyle{ \left\{ \system{#1}_{#2} \right\}_{#2 \in #3} \left\{ \system{#1}_{#4} \right\}_{#4 \in #5}} }}

\newcommand{\Sys}[1]{\ensuremath{\mathsf{Sys\left(\mathrm{#1}\right)}}}
\def\sysDimensionD{\mathsf{D}}
\newcommand{\sysDimension}[1]{\ensuremath{\sysDimensionD_{\system{#1}}}}

\newcommand{\outcomeSpace}[1]{\ensuremath{\mathsf{#1}}}
\newcommand{\outcome}[1]{\ensuremath{{#1}}}
\newcommand{\outcomeDouble}[2]{\left(\outcome{#1}, \outcome{#2} \right)}
\newcommand{\outcomeSingle}[2]{\ensuremath{ \outcome{#1} \in \outcomeSpace{#2} }}
\newcommand{\outcomeSpaceDouble}[2]{\ensuremath{\cartesianProduct{\outcomeSpace{#1}}{\outcomeSpace{#2}}}}
\newcommand{\outcomeSingleDouble}[4]{\ensuremath{ \outcomeDouble{#1}{#2} \in \outcomeSpaceDouble{#3}{#4} }}
\newcommand{\outcomeSpaceConditioned}[2]{\ensuremath{\outcomeSpace{#1}_{\outcome{#2}}}}

\newcommand{\identityTest}[1]{\mathscr{I}_{\system{#1}}}

\newcommand{\testNoDown}[1]{\ensuremath{\mathsf{#1}}}

\newcommand{\test}[2]{\ensuremath{\testNoDown{#1}_{\outcomeSpace{#2}}}}

\newcommand{\testPreparationNoDown}[1]{{\ensuremath{#1}}}
\newcommand{\testPreparation}[2]{\ensuremath{\testPreparationNoDown{#1}_{\outcomeSpace{#2}}}}

\newcommand{\testObservationNoDown}[1]{{\ensuremath{\text{#1}}}}
\newcommand{\testObservation}[2]{\ensuremath{\testObservationNoDown{#1}_{\outcomeSpace{#2}}}}

\newcommand{\testCollection}[1]{\ensuremath{\mathsf{Test\left(\mathrm{#1}\right)}}}
\newcommand{\testCollectionA}[1]{\ensuremath{\mathsf{Test\left(\system{#1}\right)}}}
\newcommand{\testCollectionAB}[2]{\ensuremath{\testCollection{\system{#1} \!\to\! \system{#2}}}}

\newcommand{\eventNoDown}[1]{\ensuremath{\mathscr{#1}}}
\newcommand{\event}[2]{\ensuremath{\eventNoDown{#1}_{\outcome{#2}}}}
\newcommand{\eventCG}[2]{\ensuremath{\eventNoDown{#1}_{\outcomeSpace{#2}}}}
\newcommand{\preparationEventNoDown}[1]{\ensuremath{#1}}
\newcommand{\preparationEvent}[2]{\ensuremath{\preparationEventNoDown{#1}_{\outcome{#2}}}}
\newcommand{\observationEventNoDown}[1]{\ensuremath{\text{#1}}}
\newcommand{\observationEvent}[2]{\ensuremath{\observationEventNoDown{#1}_{\outcome{#2}}}}
\newcommand{\observationUniqueDeterministic}{\observationEventNoDown{\uniDetEff}}

\newcommand{\probabilityEventNoDown}[1]{\ensuremath{ {#1} }}

\newcommand{\probabilityEvent}[2]{\ensuremath{ \probabilityEventNoDown{#1}_{\outcome{#2}} }}

\newcommand{\conditionedEvent}[3]{\eventNoDown{#1}^{\left(\outcome{#3}\right)}_{\outcome{#2}}}

\newcommand{\eventTest}[3]{\ensuremath{\left\{\eventNoDown{#1}_{\outcome{#2}} \right\}_{\outcomeSingle{#2}{#3}}}}
\newcommand{\observationEventTest}[3]{\ensuremath{\left\{\observationEventNoDown{#1}_{\outcome{#2}} \right\}_{\outcomeSingle{#2}{#3}}}}

\newcommand{\eventPreparationTest}[3]{\ensuremath{\left\{\preparationEventNoDown{#1}_{\outcome{#2}} \right\}_{\outcomeSingle{#2}{#3}}}}
\newcommand{\eventPreparationTestNoKet}[3]{\ensuremath{\left\{\preparationEventNoDown{#1}_{\outcome{#2}} \right\}_{\outcomeSingle{#2}{#3}}}}
\newcommand{\eventPreparationTestSystem}[4]{\ensuremath{\left\{\preparationEventNoDown{#1}_{\outcome{#2}} \right\}_{\outcomeSingle{#2}{#3}}}}
\newcommand{\eventObservationTest}[3]{\ensuremath{\left\{\observationEventNoDown{#1}_{\outcome{#2}}\right\}_{\outcomeSingle{#2}{#3}}}}
\newcommand{\eventObservationTestSystem}[4]{\ensuremath{\left\{\observationEventNoDown{#1}_{\outcome{#2}}\right\}_{\outcomeSingle{#2}{#3}}}}
\newcommand{\eventObservationTestNoBra}[3]{\ensuremath{\left\{\observationEventNoDown{#1}_{\outcome{#2}} \right\}_{\outcomeSingle{#2}{#3}}}}
\newcommand{\eventCollection}[1]{\ensuremath{\mathsf{Transf\left(\mathrm{#1}\right)}}}

\newcommand{\eventCollectionAB}[2]{\ensuremath{\eventCollection{\system{#1} \!\to\! \system{#2}}}}

\newcommand{\seqC}{\circ}
\newcommand{\sequentialComp}[2]{{#1} \seqC {#2}}
\newcommand{\sequentialEventTest}[6]{\left\{ \sequentialComp{\mathcal{#4}_{\outcome{#5}}}{\mathcal{#1}_{\outcome{#2}}} \right\}_{\outcomeSingleDouble{#2}{#5}{#3}{#6}}}

\newcommand{\paralC}{ \boxtimes }
\newcommand{\parallelComp}[2]{{#1} \paralC {#2}}
\newcommand{\parallelEventTest}[6]{\left\{ \parallelComp{\mathcal{#1}_{\outcome{#2}}}{\mathcal{#4}_{\outcome{#5}}} \right\}_{\outcomeSingleDouble{#2}{#5}{#3}{#6}}}

\newcommand{\Braid}{\mathcal{S}}

\newcommand{\permutationCollection}[1]{\ensuremath{\mathsf{Permutation\left(\mathrm{#1}\right)}}}
\newcommand{\permutationCollectionAB}[2]{\permutationCollection{\system{#1} \!\to\! \system{#2}}}

\newcommand{\St}[1]{\ensuremath{\mathsf{St\left(\system{#1}\right)}}}					

\newcommand{\StOPT}[1]{\ensuremath{\mathsf{St\left({#1}\right)}}}					    

\newcommand{\Eff}[1]{\ensuremath{\mathsf{Eff\left(\system{#1}\right)}}}					    
\newcommand{\EffR}[1]{\ensuremath{\mathsf{Eff}_{\mathbb{R}}\left(\system{#1}\right)}}	    

\newcommand{\EffOPT}[1]{\ensuremath{\mathsf{Eff\left({#1}\right)}}}					    

\newcommand{\Transf}[2]{\ensuremath{\mathsf{Transf\left(\system{#1}\!\to\!\system{#2}\right)}}}

\newcommand{\RevTransfA}[1]{\ensuremath{\mathsf{RevTransf\left(\system{#1}\right)}}}

\newcommand{\Prep}[1]{\ensuremath{\mathsf{Prep\left(\system{#1}\right)}}}

\newcommand{\Obs}[1]{\ensuremath{\mathsf{Obs\left(\system{#1}\right)}}}

\newcommand{\measurePrepare}[4]{
	\myQcircuit{
		&\s{#1}\qw&\measureD{#4}&\prepareC{#3}&\s{#2}\qw&\qw&
	}
}

\newcommand{\minimalDeterministicCausalDestroyReprep}[8]{
	\myQcircuit{
		&\s{#1}\qw&\multigate{1}{#6}&\s{#3}\qw&\measureD{\observationUniqueDeterministic}&\pureghost{}&\prepareC{#8}&\s{#4}\qw&\multigate{1}{#7}&\s{#2}\qw&\qw&
		\\
		&\pureghost{}&\pureghost{#6}&\qw&\qw&\s{#5}\qw&\qw&\qw&\ghost{#7}&
	}
}

\newcommand{\minimalDeterministicCausalDestroyReprepSequencePrime}[9]{
	\myQcircuit{
		&\s{#1}\qw&\multigate{1}{#6}&\sSequencePrime{#3}{#9}\qw&\measureD{\observationUniqueDeterministic}&\pureghost{}&\prepareC{#8}&\sSequencePrime{#4}{#9}\qw&\multigate{1}{#7}&\s{#2}\qw&\qw&
		\\
		&\pureghost{}&\pureghost{#6}&\qw&\qw&\sSequence{#5}{#9}\qw&\qw&\qw&\ghost{#7}&
	}
}

\newcommand{\genericT}[8]{
	\myQcircuit{
		&\pureghost{}&\multiprepareC{1}{#5}&\qw&\s{#4}\qw&\qw&\multimeasureD{1}{#6}&
		\\
		&\pureghost{}&\pureghost{#5}&\s{#2}\qw&\braidingSym&\s{#1}\qw&\ghost{#6}&
		\\
		&\s{A}\qw&\multigate{1}{#7}&\s{#1}\qw&\braidingGhost&\s{#2}\qw&\multigate{1}{#8}&\s{B}\qw&\qw&
		\\
		&\pureghost{}&\pureghost{#7}&\qw&\s{#3}\qw&\qw&\ghost{#8}&
	}
}

\newcommand{\transfArrow}[1]{\stackrel{#1}{\longrightarrow}}
\newcommand{\allowDisplayBreaks}[1]{
	\begingroup
	\allowdisplaybreaks
	
	#1
	
	\endgroup
}

\newcommand{\letter}{article}

\newcommand{\citetheorem}[1]{{#1}}

\makeatletter
\newcommand\footnoteref[1]{\protected@xdef\@thefnmark{\ref{#1}}\@footnotemark}
\makeatother

\newcommand{\aref}[1]{\hyperref[#1]{Appendix~\ref*{#1}}}

\newcommand{\old}[1]{\textcolor{red}{\st{#1}}}

\newcommand{\qw}[1][-1]{\ar @{-} [0,#1]}

\newcommand{\gate}[1]{*{\xy *+<.6em>{#1};p\save+LU;+RU **\dir{-}\restore\save+RU;+RD **\dir{-}\restore\save+RD;+LD **\dir{-}\restore\POS+LD;+LU **\dir{-}\endxy} \qw}

\newcommand{\measureD}[1]{*{\xy*+=+<.5em>{\vphantom{\rule{0em}{.1em}#1}}*\cir{r_l};p\save*!R{#1} \restore\save+UC;+UC-<.5em,0em>*!R{\hphantom{#1}}+L **\dir{-} \restore\save+DC;+DC-<.5em,0em>*!R{\hphantom{#1}}+L **\dir{-} \restore\POS+UC-<.5em,0em>*!R{\hphantom{#1}}+L;+DC-<.5em,0em>*!R{\hphantom{#1}}+L **\dir{-} \endxy} \qw}

\newcommand{\multimeasureD}[2]{*+<1em,.9em>{\hphantom{#2}}\save[0,0].[#1,0];p\save !C *{#2},p+LU+<0em,0em>;+RU+<-.8em,0em> **\dir{-}\restore\save +LD;+LU **\dir{-}\restore\save +LD;+RD-<.8em,0em> **\dir{-} \restore\save +RD+<0em,.8em>;+RU-<0em,.8em> **\dir{-} \restore \POS !UR*!UR{\cir<.9em>{r_d}};!DR*!DR{\cir<.9em>{d_l}}\restore \qw}

\newcommand{\multigate}[2]{*+<1em,.9em>{\hphantom{#2}} \qw \POS[0,0].[#1,0];p !C *{#2},p \save+LU;+RU **\dir{-}\restore\save+RU;+RD **\dir{-}\restore\save+RD;+LD **\dir{-}\restore\save+LD;+LU **\dir{-}\restore}
\newcommand{\ghost}[1]{*+<1em,.9em>{\hphantom{#1}} \qw}
\newcommand{\ustick}[1]{*!D!<0em,-.5em>=<0em>{#1}}

\newcommand{\Qcircuit}[1][0em]{\xymatrix @*=<#1>}
\newcommand{\node}[2][]{{\begin{array}{c} \ _{#1}\  \\ {#2} \\ \ \end{array}}\drop\frm{o} }

\newcommand{\pureghost}[1]{*+<1em,.9em>{\hphantom{#1}}}
\newcommand{\multiprepareC}[2]{*+<1em,.9em>{\hphantom{#2}}\save[0,0].[#1,0];p\save !C
  *{#2},p+RU+<0em,0em>;+LU+<+.8em,0em> **\dir{-}\restore\save +RD;+RU **\dir{-}\restore\save
  +RD;+LD+<.8em,0em> **\dir{-} \restore\save +LD+<0em,.8em>;+LU-<0em,.8em> **\dir{-} \restore \POS
  !UL*!UL{\cir<.9em>{u_r}};!DL*!DL{\cir<.9em>{l_u}}\restore}
\newcommand{\prepareC}[1]{*{\xy*+=+<.5em>{\vphantom{#1\rule{0em}{.1em}}}*\cir{l^r};p\save*!L{#1} \restore\save+UC;+UC+<.5em,0em>*!L{\hphantom{#1}}+R **\dir{-} \restore\save+DC;+DC+<.5em,0em>*!L{\hphantom{#1}}+R **\dir{-} \restore\POS+UC+<.5em,0em>*!L{\hphantom{#1}}+R;+DC+<.5em,0em>*!L{\hphantom{#1}}+R **\dir{-} \endxy}}

\newcommand{\braidingGhost}{\ghost{0em}}

\newcommand{\braidingInvId}{*+<1em,.9em>{\hphantom{0em}}; \POS[0,0]+R+<-0.98em,-0.439em>;\POS[1,0]+L+<0.98em,0.439em> **\dir{-}; \POS[0,0]*!LU{\cir<1.5em>{ur_r}};  \POS[1,0]*!RD{\cir<1.5em>{dl_l}}}

\newcommand{\braidingSym}{*+<1em,.9em>{\hphantom{0em}} \qw; \POS[0,0]+L+<0.98em,-0.439em>;\POS[1,0]+R+<-0.98em,0.439em> **\dir{-}; \POS[0,0]+R+<-0.98em,-0.439em>;\POS[1,0]+L+<0.98em,0.439em> **\dir{-}; \POS[0,0]*!RU{\cir<1.5em>{r_dr}}; \POS[0,0]*!LU{\cir<1.5em>{ur_r}};  \POS[1,0]*!RD{\cir<1.5em>{dl_l}}; \POS[1,0]*!LD{\cir<1.5em>{l_ul}}}

\newcommand{\splitterGhost}{\pureghost{0em}}

\newcommand{\splitter}{*+<1em,.9em>{\hphantom{0em}} \qw; \POS[0,0]+L+<0.98em,-0.439em>;\POS[1,0]+R+<-0.98em,0.439em> **\dir{-}; \POS[0,0]+L;\POS[0,0]+R **\dir{-}; \POS[0,0]*!RU{\cir<1.5em>{r_dr}}; \POS[1,0]*!LD{\cir<1.5em>{l_ul}}}

\newcommand{\myQcircuit}[1]{\begin{aligned} \Qcircuit @C=0.8em @R=0.8em {#1} \end{aligned} }

\newcommand{\myQcircuitSupMat}[1]{\begin{aligned} \Qcircuit @C=1em @R=1em {#1} \end{aligned} }

\usepackage{xcolor}
\usepackage{soul}

\begin{document}
\title{Measurement incompatibility is strictly stronger than disturbance}

\author{Marco Erba}
\email{marco.erba@ug.edu.pl}
\affiliation{International Centre for Theory of Quantum Technologies, Uniwersytet Gdański, ul.~Jana Bażyńskiego 1A, 80-309 Gdańsk, Poland}

\author{Paolo Perinotti}
\email{paolo.perinotti@unipv.it}
\affiliation{Universit\`a degli Studi di Pavia, Dipartimento di Fisica, QUIT Group}
\affiliation{INFN Gruppo IV, Sezione di Pavia, via Bassi 6, 27100 Pavia, Italy}

\author{Davide Rolino}
\email{davide.rolino01@universitadipavia.it}
\affiliation{Universit\`a degli Studi di Pavia, Dipartimento di Fisica, QUIT Group}
\affiliation{INFN Gruppo IV, Sezione di Pavia, via Bassi 6, 27100 Pavia, Italy}

\author{Alessandro Tosini}
\email{alessandro.tosini@unipv.it}
\affiliation{Universit\`a degli Studi di Pavia, Dipartimento di Fisica, QUIT Group}
\affiliation{INFN Gruppo IV, Sezione di Pavia, via Bassi 6, 27100 Pavia, Italy}

\begin{abstract}
  The core of Heisenberg's heuristic argument for the uncertainty principle, involving the famous $\gamma$-ray microscope \emph{Gedankenexperiment}, hinges upon the existence of measurements that irreversibly alter the state of the system on which they are acting, causing an irreducible disturbance on subsequent measurements. The argument was put forward to justify measurement incompatibility in quantum theory, namely, the existence of measurements that cannot be performed jointly---a feature that is now understood to be different from irreversibility of measurement disturbance, though related to it. In this \letter{}, on the one hand, we provide a compelling argument showing that measurement incompatibility is indeed a sufficient condition for irreversibility of measurement disturbance; while, on the other hand, we exhibit a toy theory, termed the minimal classical theory (MCT), that is a counterexample for the converse implication.
  This theory is classical, hence it does not have complementarity nor preparation uncertainty relations, and it is both Kochen-Specker and generalised noncontextual. However, MCT satisfies not only irreversibility of measurement disturbance, but also the properties of no-information without disturbance and no-broadcasting, implying that these cannot be understood \emph{per se} as signatures of nonclassicality.
\end{abstract}
	
\maketitle

\section{Introduction}
Since Heisenberg's $\gamma$-ray microscope 
\emph{Gedankenexperiment}~\cite{article:heisemberg}, the relation between measurement disturbance and the existence of pairs of observables that cannot be jointly measured has puzzled the authors that tackled quantum measurement theory~\cite{article:robertson,article:Lahti1980,PhysRevA.53.2038,article:ozawaSecond,DAriano:2003aa,article:busch2009,article:busch2004,article:teikoSimultaneous,article:heinosaari2019}. Over the years, several arguments have been proposed in favour of the fact that these two facets of quantum theory might not necessarily be equivalent. While it may seem intuitively true that the impossibility of jointly measuring two observables necessarily implies measurement disturbance, a proof of this fact has never been given in a theory-independent fashion. On the other hand, it is not even intuitive whether or not the converse implication should hold true. Indeed, no conclusive argument has been given so far in favour or against the latter. The main difficulty in this direction is that quantum theory (QT) exhibits both features, while classical theory (CT) exhibits none of them.

In order to understand the logical relation between incompatibility of observations\footnote{From now on, we will use the word \emph{observation} to denote a measurement where the output system is discarded, or simply disregarded. Notice that any measurement determines an observation, corresponding to the measurement itself followed by a discard of the output system.} and irreversibility of measurement disturbance, one needs to move outside the limited scenarios of QT and CT, broadening the perspective to the wider context of general probabilistic theories. In Ref.~\cite{article:heinosaari2019}, the authors exhibit a theory where there are some measurements that cause irreversible disturbance, while the corresponding observations are compatible with all the remaining ones. However, so far no theory has been exhibited such that \emph{all} of its observations are compatible, and yet their measurements cause irreversible disturbance, thus decoupling irreversibility from incompatibility.

In the present \letter{}, we address the above question in the framework of operational probabilistic theories (OPTs)~\cite{article:probTheoriesPurif,book:quantumFromPrinciples,book:quantumTheoryFromFirstPrinciples,article:BCT}. These are generic theories of information, including QT and CT as particular cases, but also encompassing a wealth of toy theories sharing the same basic compositional structures for systems and processes. This framework is the appropriate one to seek general arguments about the logical dependency of different properties that physical theories may exhibit. Indeed, the operational-probabilistic framework has been devised in order to survey general physical theories ``from the outside.'' Relevant quantum properties, such as entanglement and contextuality, have been then investigated in a similar generalised  scenario. For instance, in Ref.~\cite{article:shelby} entanglement is established as an inevitable feature of any theory superseding CT while admitting emergent classicality. Furthermore,  in Ref.~\cite{article:spekkensSecond} a contextuality witness is deduced in terms of the functional form of an uncertainty relation, thus pinpointing some aspects of quantum uncertainty that may constitute genuine evidence of nonclassicality.

The properties of interest for our analysis are as follows: (i) incompatibility of observations, along with uncertainty, and (ii) irreversibility of measurement disturbance. For case (i), a family of theories being of interest is the one of \emph{epistemically restricted} classical phase-space theories~\cite{PhysRevA.75.032110,Spekkens2016}, i.e., operational theories with a phase space~\cite{PhysRevLett.128.040405} that is possibly discrete~\cite{PhysRevA.70.062101}, where some restriction in the spirit of Heisenberg's uncertainty principle singles out a minimal volume in phase space that can be identified by a pure state~\cite{PhysRevA.86.012103}. In the present \letter{}, we discuss a family of theories motivated by point (ii), i.e., theories where irreversibility of the measurement disturbance holds.

In detail, in this \letter{} we show that the existence of pairs of incompatible observations---a property that we term \emph{incompatibility} for short---is a strictly stronger condition than the existence of operations that irreversibly disturb the state of the system on which they act---that we name \textdef{irreversibility}.  In order to achieve this, we will prove that the former property implies the latter one, while exhibiting a toy theory that violates the converse implication. The counterexample consists in a theory---that we call minimal classical theory (MCT)---where all observations are compatible, but any measurement irreversibly alters the state of the system. This theory is obtained from CT by restricting to the bone the set of operations one is allowed to perform on a system. The states of systems of MCT being classical, this theory also represents a proof that, contrarily to what is normally believed, the disturbance action caused by the interaction with a system is not a characteristic property of the quantum world. Moreover, this result is complementary to the one of Ref.~\cite{article:spekkensSecond}, in that MCT has irreversibility while having no incompatibility of observations---hence no uncertainty thereon---thus being clearly also Kochen-Specker noncontextual.
Furthermore, MCT, being embeddable~\cite{PRXQuantum.2.010331} into classical theory, is generalized-noncontextual according to Ref.~\cite{schmid2020structure}. We observe that, whereas the pure states of every system can be jointly and perfectly discriminated within the theory, MCT satisfies the property of \emph{no-information without disturbance}~\cite{article:busch2009,book:quantumTheoryFromFirstPrinciples,article:informationDisturbance,article:heinosaari2019}. Furthermore, \emph{generalised no-broadcasting}~\cite{PhysRevLett.99.240501,barnum2006cloning}---and, as a particular case, \emph{no-cloning}---hold in the theory.

\section{Framework}

We now sketch the framework of operational probabilistic theories which is here leveraged. An OPT is meant as a theory of systems and their processes. The probabilistic aspect consists in rules to assess the probability of events in any network of processes occurring on a given set of systems. 
In detail, a generic OPT \OPT{} has a collection of \textdef{systems} and of \textdef{tests} thereon. Systems are denoted by capital roman letters $\system{A}, \system{B},\ldots \in \Sys{\OPT{}}$. As an example, every system in QT corresponds to a complex Hilbert space. Tests, denoted as $\test{T}{X} = \eventTest{T}{x}{X}$, represent the experiments that one can perform, acting on a given input system \system{A}, and obtaining the output system \system{B}. The class of tests with input system $\system A$ and output system $\system B$ is denoted by $\testCollectionAB{A}{B}$.
Every test consists in a collection of possible \emph{transformations} $\eventNoDown{T}_{x}$, labelled by the possible outcomes $x\in\outcomeSpace{X}$ of the experiment. A finite \textdef{outcome space} (\outcomeSpace{X}) is associated with each test. 
The class of transformations with input $\system A$ and output $\system B$ is denoted by $\eventCollectionAB{A}{B}$.
Referring again to QT, tests are \emph{quantum instruments}, and transformations are \emph{quantum operations}. For example, in a Stern--Gerlach experiment the test that models the action of the magnetic field is of the form $\testNoDown{T}_{\left( \uparrow, \downarrow \right)} = \left\{ \eventNoDown{T}_{\uparrow}, \eventNoDown{T}_{\downarrow} \right\}$, where the two transformations $ \eventNoDown{T}_{\uparrow}$ and $ \eventNoDown{T}_{\downarrow}$ represent the two occurrences in which the system collapses into a state with spin up or down, respectively. A transformation associated with a test whose outcome space has just one element is called \emph{deterministic}. A deterministic transformation does not provide information (the associated test has a unique outcome, occurring with certainty), and can represent, e.g., the evolution of an open system. In QT, a deterministic transformation is a quantum channel.

The first main feature of tests is that they can be performed in a sequence, where a sequence can be defined whenever the input of the subsequent test is the same as the output of the preceding test. Tests (and transformations) will be drawn as boxes, and this makes the representation of a sequence of tests (transformations) more intuitive,
\begin{align*}
	 \myQcircuit{
		&\s{A}\qw&\gate{\test{T}{X}}&\s{B}\qw&\gate{\test{G}{Y}}&\s{C}\qw&\qw&
	}=\quad
	\myQcircuit{
		&\s{A}\qw&\gate{\sequentialComp{\test{G}{Y}}{\test{T}{X}}}&\s{C}\qw&\qw&
	},
\end{align*}
where $\sequentialComp{\test{G}{Y}}{\test{T}{X}}=\sequentialEventTest{T}{x}{X}{G}{y}{Y}$. 

A second defining structure of OPTs is \emph{parallel composition}, that allows one to combine any pair of systems $\system A$ and $\system B$ in a \emph{composite system} $\system{AB}$. Given a composite system, moreover, one can independently apply tests $\test{T}{X}$ and $\test{G}{Y}$ on the two components. The resulting test is the parallel composition $\parallelComp{\test{T}{X}}{\test{G}{Y}}$ that is drawn as follows, 
\begin{align*}
	\myQcircuit{
		&\s{A}\qw&\gate{\test{T}{X}}&\s{C}\qw&\qw&
		\\
		&\s{B}\qw&\gate{\test{G}{Y}}&\s{D}\qw&\qw&
	} &= \quad 
	\myQcircuit{
		&\s{AB}\qw&\gate{\parallelComp{\test{T}{X}}{\test{G}{Y}}}&\s{CD}\qw&\qw&
	},
\end{align*}
where $\parallelComp{\test{T}{X}}{\test{G}{Y}}=\parallelEventTest{T}{x}{X}{G}{y}{Y}$. Both sequential and parallel composition are associative.

A special kind of test consists in the \emph{preparation} of a system $\system A$. These tests are called \textdef{preparation tests}, and their class deserves a dedicated symbol: $\Prep{A}$. The possible
transformations of a preparation test are \emph{states}, that can be denoted as $\rketSystem{\preparationEventNoDown{\rho}}{A}\in \St{A}$. Similarly, a special class of tests is that representing measurements after which the system $\system A$ is 
destroyed, discarded, or just neglected.
As mentioned in Ref.~\cite{Note1}, these tests are called \textdef{observation tests} or \emph{observations} for short, and their set is denoted by $\Obs{A}$. The transformations $\rbraSystem aA$ of an observation test are called \textdef{effects}, and their set is denoted by $\Eff{A}$. Observation tests are the generalisation of positive operator-valued measures (POVMs) of QT to generic theories. We will draw states and effects as
\begin{align*}
\myQcircuit{\prepareC{\preparationEventNoDown{\rho}}&\s{A}\qw&\qw}\ ,\quad\myQcircuit{
		&\s{A}\qw&\measureD{\observationEventNoDown{a}}}\  , 
\end{align*}
respectively. Preparation (observation) tests can be regarded as special tests whose input (output) system is \emph{trivial}. The (unique) trivial system is denoted by $\trivialSystem{}$. From an operational point of view this system represents ``nothing the theory cares to describe''~\cite{book:quantumFromPrinciples}. The trivial system behaves as a unit for parallel composition: $\system{AI}=\system{IA}=\system A$. \par

For every system $\system A$ of the theory, we require the existence of a deterministic transformation $\identityTest{A}$---called \textdef{identity}---representing ``doing nothing'' on the system, i.e., such that $\identityTest{A} \eventNoDown{T} = \eventNoDown{T} \identityTest{B} = \eventNoDown{T}$, for every transformation $\eventNoDown{T} \in \eventCollectionAB{A}{B}$. A transformation $\eventNoDown{T}\in \eventCollectionAB{A}{B}$ is \textdef{reversible} if there exists $\inverse{\eventNoDown{T}}\in \eventCollectionAB{B}{A}$ such that $\eventNoDown{T} \inverse{\eventNoDown{T}} =\identityTest{B}$ and $ \inverse{\eventNoDown{T}} \eventNoDown{T} = \identityTest{A}$. Moreover, for every pair of systems $\system A$ and $\system B$, we require the existence of the swap operation representing the exchange of the two systems, i.e.,
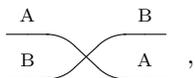
\begin{equation*}
	\myQcircuit{
		&\s{A}\qw&\braidingSym&\s{B}\qw&\qw&
		\\
		&\s{B}\qw&\braidingGhost&\s{A}\qw&\qw&
	},
\end{equation*}
which is a deterministic reversible transformation. Tests and transformations  can slide along the crossed wires through a swap.

In an OPT every circuit that starts with a preparation test and ends with an observation test represents a probability distribution for the transformations in the circuit. For example,
\begin{equation*}
	\myQcircuit{
&\prepareC{\preparationEvent{\rho}{x}}&\s{A}\qw&\gate{\event{T}{y}}&\s{B}\qw&\measureD{\observationEvent{a}{Z}}&
	} \mathDef \probabilityEventNoDown{p}\left( \outcome{x}, \outcome{y}, \outcome{z} \vert \testPreparation{\rho}{X}, \test{T}{Y}, \testObservation{a}{Z} \right).
\end{equation*}
An agent performing a test can discard information regarding the outcome. Correspondingly, in an OPT we require that for every test $\test{T}{X}$ and every disjoint partition $\{Z_y\}_{y\in\outcomeSpace Y}$ of the outcome space $\outcomeSpace X$ there exists the test $\testNoDown{T}^{'}_{\outcomeSpace{Y}}$ representing the same operation, where the outcome $y\in\outcomeSpace Y$ stands for ``the outcome of the test $\test{T}{X}$ belongs to $Z_y$.'' The transformation $\eventNoDown T'_y=\sum_{x\in Z_y}\eventNoDown T_x$ is called \textdef{coarse-grained transformation}. 
Obviously, given a test \test{T}{X} the full coarse graining $\eventCG{T}{X} = \sum_{\outcome{x} \in \outcomeSpace{X}} \event{T}{x}$ is deterministic.

The above operational apparatus naturally gives rise to a linear-space structure, where transformations are embedded in real vector spaces.\\ 
The spaces of transformations and tests within a theory are required to be Cauchy complete. This follows from the idea that if, within a given theory, there is a procedure to prepare a transformation (or a test) with arbitrary precision, then it is natural to assume that the latter is an ideal transformation (or test) to be included in the theory. This requirement is particularly relevant for the present work, since it allows one to distinguish two theories that share similar building blocks, but where, nonetheless, the operational procedures of one of them cannot be approximated by those of the other one. Importantly, this will allow us to prove that MCT is strictly different from ordinary CT.

For more details about the framework, we refer the reader to Refs.~\cite{book:quantumTheoryFromFirstPrinciples,article:BCT}.

\section{Definitions}\label{sec:defs}
In the following we will consider only \textit{causal} OPTs. These are theories where any system admits of a unique deterministic effect, denoted by $\rbraSystem{\observationUniqueDeterministic}{A} \in \Eff{A}$. This condition is equivalent to the property that the probability distributions of preparation tests do not depend on the choice of the observation test at their output---a property known as \emph{no-signaling from the future}. This property also implies spatial no-signaling, namely, the property of \emph{no-signaling without interaction}~\cite{book:quantumTheoryFromFirstPrinciples,article:BCT}.

We now introduce the notion of compatibility of observation tests, which will play a central role in our results. The definition is borrowed from a wide literature on the subject~(see, e.g., Refs.~\cite{ludwig,Busch_2013,PhysRevLett.124.120401}), where compatibility is ubiquitously identified with joint measurability. In precise terms, we say that the observation tests $\testObservation{a}{X}\in\Obs{A}$ and $\testObservation{b}{Y}\in\Obs{A}$ are \emph{compatible} if there exists a third test $\testObservation{c}{X\times Y}\in\Obs{A}$ such that
\begin{align*}
	&\myQcircuit{
		&\s{A}\qw&\measureD{\observationEvent{a}{x}}& 
	}  = \quad \sum_{\outcomeSingle{y}{Y}} \myQcircuit{
		&\s{A}\qw&\measureD{\observationEvent{c}{\left(x,y\right)}}&
	} \quad \forall \outcomeSingle{x}{X}, \\	
	&\myQcircuit{
		&\s{A}\qw&\measureD{\observationEvent{b}{y}}& 
	}  = \quad \sum_{\outcomeSingle{x}{X}} 
	\myQcircuit{
		&\s{A}\qw&\measureD{\observationEvent{c}{\left(x,y\right)}}&
	}  \quad \forall \outcomeSingle{y}{Y}.
\end{align*}
Accordingly, we will say that a theory has \emph{incompatibility} if it admits of a system $\system A$ and a pair of observation tests for $\system A$ that are not compatible.

In order to determine whether an OPT exhibits tests with a disturbance in the sense of Heisenberg, i.e., when an OPT has \textdef{irreversibility}, we require the existence of at least a test that irreversibly alters the state of the system on which it acts. In this way, we are stating that these operations set a direction for the arrow of time, in analogy with the second law of thermodynamics. Accordingly, we say that a test is \textdef{intrinsically irreversible} if its occurrence precludes the possibility of implementing some other test~\cite{article:incompatibility} on the same input system. Notice that in general one can implement a test using ancillary systems, and our definition allows one to postprocess them along with the output system. The precise definition of intrinsic irreversibility is then the following. We say that the test $\test{A}{X}\in\testCollectionAB{A}{B}$ is intrinsically irreversible if it \emph{excludes} some other test $\test{B}{Y}\in\testCollectionAB{A}{C}$~\cite{article:incompatibility}, i.e., there exists a test $\test{B}{Y}\in\testCollectionAB{A}{C}$ such that, for every $\test{C}{Z}\in\testCollectionAB{A}{BE}$ and every disjoint partition $\{S_x\}_{x\in\outcomeSpace X}$ of $\outcomeSpace Z$ with
\begin{equation}
	\label{eq:decompa}
	\myQcircuit{
		&\s{A}\qw&\gate{\eventNoDown{A}_x}&\s {B}\qw&\qw&
	} = \sum_{z\in S_x}
	\myQcircuit{
		&\s{A}\qw&\multigate{1}{\eventNoDown{C}_z}&\qw&\s {B}\qw&\qw&
		\\
		&\pureghost{}&\pureghost{\eventNoDown{C}_z}&\s{E}\qw&\measureD{e}}
\end{equation}
there exists no postprocessing $\mathsf{P}^{(z)}_{\mathsf Y}\in\testCollectionAB{BE}{C}$ such that
\begin{equation}
	\label{eq:decompb}
	\myQcircuit{
		&\s{A}\qw&\gate{\eventNoDown{B}_y}&\s {C}\qw&\qw&
	} = \quad \sum_{z\in \outcomeSpace Z}
	\myQcircuit{
		&\s{A}\qw&\multigate{1}{\eventNoDown{C}_z}&\s {B}\qw&\multigate{1}{\eventNoDown{P}^{(z)}_y}&\s{C}\qw&\qw&
		\\
		&\pureghost{}&\pureghost{\eventNoDown{C}_z}&\s{E}\qw&\ghost{\eventNoDown{P}^{(z)}_y}&\pureghost{}&
	}
\end{equation}
A first result that we can prove is that a test is intrinsically irreversible if and only if it excludes the  identity test. Indeed, if this is the case, the above definition holds choosing $\test{B}{Y}=\{\identityTest{A}\}$. On the other hand, by contradiction, if $\test{A}{X}$ can be postprocessed to the identity test, then it can be postprocessed to any other test. The detailed proof of the previous statement can be found in \aref{app:irreversibility}.

In the light of the above discussion, we will say that a theory has \emph{irreversibility} if it admits of a test that is intrinsically irreversible. Notice that, according to our definition, in QT---where all channels admit of a unitary dilation---no channel is intrinsically irreversible. On the other hand, almost all quantum tests are intrinsically irreversible. Irreversibility thus stems, at least in QT, from the very extraction of information in a measurement.

\section{Incompatibility versus irreversibility} \label{sec:irr_implies_inc}

\subsection{Incompatibility implies irreversibility}

We can now prove the first of our two main results: The existence of incompatible observation tests implies that the theory has irreversibility.

Given two observation tests $\eventObservationTestSystem{a}{x}{X}{A}$, $\eventObservationTestSystem{b}{y}{Y}{A} \in \Obs{A}$, if they do not exclude each other, then one can straightforwardly prove that they are compatible. Hence, incompatibility is sufficient for irreversibility. Actually, in any operational theory with nontrivial systems, incompatibility is sufficient also for intrinsic irreversibility of some tests with nontrivial output. In order to prove this, we show that there always exist two tests $\eventTest{T}{x}{X} \in \testCollectionAB{A}{B}$ and $\eventTest{G}{y}{Y} \in \testCollectionAB{A}{C}$ such that
\begin{equation}
	\begin{aligned}
		\label{eqt:formula1}
		&\myQcircuit{
			&\s{A}\qw&\measureD{\observationEvent{a}{x}}&
		} = \quad \myQcircuit{
			&\s{A}\qw&\gate{\event{T}{x}}&\s{B}\qw&\measureD{\observationUniqueDeterministic}&
		} \quad \forall \outcomeSingle{x}{X}, \\[10pt] &\myQcircuit{
			&\s{A}\qw&\measureD{\observationEvent{b}{y}}&
		} = \quad \myQcircuit{
			&\s{A}\qw&\gate{\event{G}{y}}&\s{C}\qw&\measureD{\observationUniqueDeterministic}&
		} \quad \forall \outcomeSingle{y}{Y}.
	\end{aligned}
\end{equation}
Indeed, in every nontrivial OPT and for every system $\system{A}$, it is always possible to choose a measure-and-prepare test
\begin{equation*}
	\myQcircuit{
		&\s{A}\qw&\gate{\event{T}{x}}&\s{B}\qw&\qw&
	} = \quad \myQcircuit{
		&\s{A}\qw&\measureD{\observationEvent{a}{x}}&\prepareC{\preparationEventNoDown{\rho}}&\s{B}\qw&\qw&
	} \quad \forall\outcomeSingle{x}{X},
\end{equation*}
where $\rketSystem{\preparationEventNoDown{\rho}}{B} \in \St{B}$ is an arbitrary deterministic state of a nontrivial system $\system{A}$, and analogously for $\eventObservationTestSystem{b}{y}{Y}{A}$.
We conclude by observing that, if by contradiction either test $\test{T}{X}$ or $\test{G}{Y}$ does not exclude the other, the observation tests $\sequentialComp{\rbraSystem{\observationUniqueDeterministic}{B}}{\test{T}{X}}$ and $\sequentialComp{\rbraSystem{\observationUniqueDeterministic}{C}}{\test{G}{Y}}$ are compatible~\cite{article:incompatibility}, contradicting the hypothesis.

Then, whenever a theory has incompatibility, there must exist at least a pair of tests with nontrivial output that exclude each other,
thus being intrinsically irreversible. In summary, \emph{incompatibility of observations implies irreversibility of measurement disturbance}.

\subsection{Irreversibility does not imply incompatibility}
We now proceed to prove the second main result, by exhibiting a toy theory called minimal classical theory (MCT) that has irreversibility but no incompatibility. This theory is obtained by restricting the sets of allowed transformations and tests of CT, while keeping its sets of states and effects untouched. More in detail, the only allowed tests (and consequently transformations) are the ones that can be obtained combining preparation tests and observation tests with the identity and swap operations (and limits of sequences of tests thereof).

MCT is an instance of a broader family of OPTs that can be analogously obtained: Starting from an OPT, one can build its minimal version by only allowing preparation tests and observation tests, permutations of systems, and arbitrary compositions or limits thereof. These theories are called \emph{minimal} OPTs. In \aref{app:mopt}, the formal definition of this family of theories is presented together with a series of results characterising their transformations and their properties. As for MCT, the definition and formalisation of the results discussed below are presented in \aref{app:mct}.

Causal classical theories are here defined as OPTs where the state spaces are simplexes whose vertices (pure states) are jointly perfectly discriminable~\cite{article:classicalEntanglement,article:BCT}. We can now review some aspects of MCT, actually referring to results that hold for arbitrary minimal OPTs. The tests of a minimal OPT---with the exclusion of (some of) the limit tests---are of the form
\begin{equation}
\label{eqt:jelly}
	\genericT{\prim{A}}{\prim{B}}{E}{C}{\eventPreparationTestNoKet{\rho}{x}{X}}{\eventObservationTestNoBra{a}{y}{Y}}{\Braid^{\left(1\right)}}{\Braid^{\left(2\right)}},
\end{equation}
where $\eventPreparationTestSystem{\rho}{x}{X}{C\prim{B}} \in \Prep{C\prim{B}}$ and $\eventObservationTestSystem{a}{y}{Y}{C\prim{A}} \in \Obs{C\prim{A}}$ are generic preparation tests and observation tests, and $\Braid^{\left(1\right)}$ and $\Braid^{\left(2\right)}$ are generic permutations (see \aref{app:permutations}).\footnote{More precisely, systems $\system A$ and $\system B$ are generally composite $\system A=\system{A_1A_2\ldots A_j}$, $\system B=\system{B_1B_2\ldots B_k}$, and $\Braid^{\left(1\right)}$ permutes the subsystems of $\system A$, while $\Braid^{\left(2\right)}$ permute those of system $\system B$.} Notice that there is some degree of arbitrariness in the choice of the systems $\system{\prim{A}}$, $\system{\prim{B}}$, $\system{C}$, $\system{E}$, and in some cases they can be taken as the trivial system \trivialSystem. 
As a consequence of the realisation scheme of tests in Eq.~\eqref{eqt:jelly}, MCT is such that the identity transformation is atomic---i.e., every test whose full coarse graining is equal to $\identityTest{}$ must be of the form $\{p_x\identityTest{}\}_{x\in\outcomeSpace X}$, with $\{p_x\}_{x\in\outcomeSpace X}$ a probability distribution---for every one of its systems.

The proof of the preceding property proceeds as follows. First, suppose that there is some test that decomposes $\identityTest{A}$. This test is the limit of some sequence $\test{T}{X}^{(n)}$ of tests of the form~\eqref{eqt:jelly}. The important fact here is that the arbitrary systems $\system{A}'_{n}$, $\system{B}'_{n}$, $\system{E}_{n}$, as well as the permutations $\Braid^{\left(1\right)}_n$ and $\Braid^{\left(2\right)}_n$, for the tests $\test{T}{X}^{(n)}$ in the sequence can be taken to be independent of $n$. Then, the full coarse graining of the limit test---that coincides with the limit of the sequence of full coarse grainings---is of the form of Eq.~\eqref{eqt:jelly} where the observation test $\eventObservationTestNoBra{a}{y}{Y}$ reduces to the deterministic effect $\rbraSystem{\observationUniqueDeterministic}{\system{C}\system{A}'} \in \Eff{C\prim{A}}$.
Since in our case $\system A\equiv\system B$, it must also be $\system A'\equiv\system B'$. Moreover, since the overall transformation must be the identity, one can easily check that it must be $\Braid^{\left(2\right)}=[\Braid^{\left(1\right)}]^{-1}$. In summary, one must have
\begin{align}
	\label{eqt:detJelly}
	&\minimalDeterministicCausalDestroyReprep{A}{A}{\prim{A}}{\prim{A}}{E}{\Braid^{\left(1\right)}}{\inverse{\Braid^{\left(1\right)}}}{\preparationEventNoDown{{\rho}}} \!\!\!\!=\ \  \myQcircuit{
		&\s{A}\qw&\qw&
	},
\end{align}
and finally, inverting the permutations on both sides, we end up with
\begin{align*}
	& \myQcircuit{
		&\s{A'}\qw&\measureD{\observationUniqueDeterministic}&\pureghost{}&\prepareC{\preparationEventNoDown{{\rho}}}&\s{A'}\qw&\qw&
		\\
		&\qw&\qw&\s{E}\qw&\qw&\qw&\qw&
	} = \quad \myQcircuit{
		&\qw&\s{A'E}\qw&\qw&\qw&
	},
\end{align*}
which then requires $\system{A}' \equiv\trivialSystem$. By the stability of the systems $\system{A}'_{n}$, $\system{B}'_{n}$, $\system{E}_{n}$, also the sequence of tests of the form~\eqref{eqt:jelly} converging to our decomposition of $\identityTest{A}$ must have trivial systems $\system A'\equiv\system B'\equiv\system I$. As a consequence, all such tests must contain transformations proportional to $\identityTest{A}$, and so must the limit test.

We finally prove that every theory admitting at least a system of dimension greater than $1$ for which the identity transformation is atomic, and MCT as a particular case, has irreversibility. This is shown by contradiction. Suppose that a theory has no intrinsically irreversible tests. Then any test $\test{A}{X}\in\testCollectionAB{A}{B}$ does not exclude the identity and is achievable via a test $\eventTest{C}{z}{Z}\in\testCollectionAB{A}{BE}$ such that Eq.~\eqref{eq:decompb} holds with $\eventNoDown{B}_y$ replaced by the identity $\identityTest{A}$. Suppose \system{A} is a system of dimension greater than one where the identity transformation is atomic, due to this, one has $\rbraSystem{e}{BE}\event{C}{z}=p_z\rbraSystem{e}{A}$, which means that the observation test associated to $\test{A}{X}\in\testCollectionAB{A}{B}$ is of the 
form $\{p_z \observationUniqueDeterministic\}_{z\in\outcomeSpace Z}$, namely it is trivial. Since this is true for every test, all observation tests of system \system{A} must be trivial, which is possible only if this system has dimension equal to $1$, thus reaching a contradiction. On the other hand, MCT admits systems of dimension greater than $1$. Therefore, since for every system of MCT the identity transformation is atomic, it has irreversibility. To conclude the argument, it is sufficient to observe that MCT does not have incompatibility, since it has the same observation tests as CT, where all the observations are compatible. A detailed proof of the previous statement is given in \aref{app:mct:compatibility}.\\
We conclude by observing that the fact that the identity transformation is atomic for any of MCT's systems sets this theory apart from CT. In fact, the latter theory satisfies the opposite property; any system of the theory always admits a test that is a nontrivial decomposition of the identity. Hence MCT's test set is \emph{strictly} contained in CT's one.

\section{Discussion} In this \letter{} we have proven that, in a general theory of physical systems, the presence of incompatible observations implies the existence of tests that are intrinsically irreversible, but the reverse does not hold. The counterexample is given in terms of a fully fledged OPT, that we named MCT, whose state spaces are simplexes. Incidentally, in Ref.~\cite{art:plavala_simplex} it is proven that, under the no-restriction hypothesis, the compatibility of all observation tests is equivalent to having simplicial state spaces. At any rate, simpliciality is sufficient for a theory to exhibit full compatibility of observations, and yet, remarkably, it does not preclude the presence of irreversible disturbance, as we have shown here. Indeed, the fact that a theory such as MCT exists is not straightforward.

Notice that it is reasonable to expect that MCT is not the unique theory with full compatibility of observations and irreversibility. Moreover, it is not even clear that such a theory must be simplicial, as suggested e.g., in Ref.~\cite{Buscemi2023unifyingdifferent}, where a possible example is sketched of a theory made of quantum systems whose unique allowed observation test corresponds to an informationally complete POVM.

The notion of intrinsic irreversibility has been also introduced in an operational framework, and characterised as the existence of 
tests that cannot be postprocessed to the identity---not even with access to arbitrary ancillary systems. The consequent notion of irreversibility---i.e., the property of a theory with an intrinsically irreversible test---is very restrictive, and one may conjecture that it lies at the origin of thermodynamic irreversibility. The analysis of this hypothesis will be the subject of future studies.

The toy theory presented here, exhibiting irreversibility but also full compatibility of observations, can be used to compare other features which are beyond the scope of this \letter{} as well. For example, MCT establishes that classicality is not sufficient for a theory to have full compatibility of tests, where the latter is defined according to Refs.~\cite{article:incompatibility,Buscemi2023unifyingdifferent}. Moreover, MCT satisfies not only no-information without disturbance as follows from the atomicity of the identity transformation~\cite{article:informationDisturbance},but also generalised no-broadcasting~\cite{PhysRevLett.99.240501,barnum2006cloning}, as follows from the form of the channels of the theory. In particular, as a special case of the latter property, MCT also satisfies no-cloning. Accordingly, MCT represents the evidence that the properties of no-information without disturbance and of no-broadcasting are \emph{not} signatures of nonclassicality \emph{per se}. 

We highlight that the fact that MCT satisfies no-broadcasting is not in contradiction with the results of Refs.~\cite{PhysRevLett.99.240501,barnum2006cloning}.
Indeed, one of the underlying assumptions of the latter works is the possibility of having classical control on outcomes, or, in other words, the possibility of choosing which test to perform conditionally on which outcome has occurred in a preceding test. However, such an assumption was not made in the present work. This establishes that classicality---understood as the joint perfect distinguishability of the pure states---is not in itself a sufficient condition for broadcasting. Now, it would be interesting to determine under which assumptions classicality entails the possibility of generalised broadcasting, e.g., by determining whether the sole addition of the above-mentioned conditional tests would be sufficient.\footnote{We remark that the property of classical control on outcomes is actually quite strong, implying e.g., the causality principle (see \hyperref[sec:defs]{Section~\ref{sec:defs}} and Ref.~\cite{book:quantumTheoryFromFirstPrinciples}).}

As a final remark, we observe that MCT does not have complementarity nor preparation uncertainty relations, such as Robertson's ones~\cite{article:robertson}, which shows that those are not implied by irreversibility, just as incompatibility is not. It is yet unknown whether or not the converse of the above statement holds.

\bigskip


\section{Acknowledgements}
P.P.~acknowledges financial support from European Union - Next Generation EU through the PNNR MUR Project No.~PE0000023. A.T.~acknowledges the financial support of Elvia and Federico Faggin Foundation (Silicon Valley Community Foundation Project ID No.~2020-214365). M.E.~acknowledges financial support from the National Science Centre, Poland (Opus project, Categorical Foundations of the Non-Classicality of Nature, Project No.~2021/41/B/ST2/03149).

\widetext
\appendix

\section{Incompatibility and intrinsic irreversibility}\label{app:irreversibility}

In the present appendix we provide the formal proof that whenever a test does not exclude the identity, then it does not exclude any other test.

\begin{lemma}
	\label{lem:OPT:instr:incompatibility:weak:weakEquiv:1}
	For any two given tests $\test{A}{X} = \eventTest{A}{x}{X} \in \testCollectionA{A}$ and $\test{B}{Y} = \eventTest{B}{y}{Y} \in \testCollectionAB{A}{B}$, if $\test{A}{X}$ does not exclude the identity, then it also does not exclude $\test{B}{Y}$.
\end{lemma}
\begin{proof}	
	Writing the non-exclusion relation of $\test{A}{X}$ with the identity as
	\begin{align*}
		& \myQcircuit{
			&\s{A}\qw&\gate{\event{A}{x}}&\s{A}\qw&\qw&
		} = \sum_{z\in \outcomeSpaceConditioned{S}{x}}
		\myQcircuit{
			&\s{A}\qw&\multigate{1}{\event{C}{z}}&\s{A}\qw&\qw&\qw&
			\\
			&\pureghost{}&\pureghost{\event{C}{z}}&\s{\prim{A}}\qw&\measureD{\observationUniqueDeterministic}&
		}, \\[10pt]
		& \myQcircuit{
			&\s{A}\qw&\qw&
		} = \sum_{z}
		\myQcircuit{
			&\s{A}\qw&\multigate{1}{\event{C}{z}}&\s{A}\qw&\multigate{1}{\conditionedEvent{P}{y}{z}}&\s{A}\qw&\qw&\qw&
			\\
			&\pureghost{}&\pureghost{\event{C}{z}}&\s{\prim{A}}\qw&\ghost{\conditionedEvent{P}{l}{z}}&\s{\secondE{A}}\qw&\measureD{\observationUniqueDeterministic}&
		},
	\end{align*}
	the result then follows by taking
	\begin{align*}
		\myQcircuit{
			&\s{A}\qw&\gate{\event{B}{y}}&\s{B}\qw&\qw&
		} & = \sum_{z}
		\myQcircuit{
			&\s{A}\qw&\multigate{1}{\event{C}{z}}&\s{A}\qw&\multigate{1}{\conditionedEvent{P}{y}{z}}&\s{A}\qw&\gate{\event{B}{y}}&\s{B}\qw&\qw&
			\\
			&\pureghost{}&\pureghost{\event{C}{z}}&\s{\prim{A}}\qw&\ghost{\conditionedEvent{P}{y}{z}}&\s{\secondE{A}}\qw&\measureD{\observationUniqueDeterministic}&
		} \\[10pt]
		& = \sum_{z} \myQcircuit{
			&\s{A}\qw&\multigate{1}{\event{C}{z}}&\s{A}\qw&\multigate{1}{\event{P'}{y}^{(z)}}&\s{B}\qw&\qw&\qw&
			\\
			&\pureghost{}&\pureghost{\event{C}{z}}&\s{\prim{A}}\qw&\ghost{\event{P'}{y}^{(z)}}&\s{\secondE{A}}\qw&\measureD{\observationUniqueDeterministic}&
		},
	\end{align*}
	where we defined the new postprocessing $\mathsf{P'}^{(z)}_{\mathsf Y}:=\mathsf{B}_{\mathsf Y}\circ\mathsf{P}^{(z)}_{\mathsf Y}\in\testCollectionAB{AA'}{BA''}$.
\end{proof}

\section{Permutations and their properties}\label{app:permutations}
In the present appendix we will discuss the particular set of reversible transformations called permutations.

\begin{definition}[Set of permutations] 
	\label{def:permutations}
	The \textdef{set of permutations}, whose representatives will be indicated with $\Braid$, is defined as the equivalence class of transformations which are obtained by parallel and sequential composition of swap and identity transformations.
\end{definition} 

The above defined transformations on bipartite systems satisfy the following characterisation theorem:

\begin{theorem}[General form of permutations on bipartite systems]
	\label{thm:OPT:symmetric:permutations:generalForm}
	In every OPT for any permutation acting on a bipartite system, there exist suitable systems $\system{A}'$, $\system{B}'$, $\system{A}''$, $\system{B}''$, and transformations $\Braid_{1}$, $\Braid_{2}$, $\Braid_{3}$, $\Braid_{4}$ such that
	\begin{equation}
		\label{eqt:OPT:permutation:symm:formula}
		\myQcircuitSupMat{
			&\s{A}\qw&\multigate{1}{\Braid}&\s{C}\qw&\qw&
			\\
			&\s{B}\qw&\ghost{\Braid}&\s{D}\qw&\qw&
		} =
		\myQcircuitSupMat{
			&\s{A}\qw&\multigate{1}{\Braid_{3}}&\qw&\s{\prim{A}}\qw&\qw&\multigate{1}{\Braid_{4}}&\s{C}\qw&\qw&
			\\
			&\pureghost{}&\pureghost{\Braid_{3}}&\s{\secondE{A}}\qw&\braidingSym&\s{\prim{B}}\qw&\ghost{\Braid_{4}}\qw&
			\\
			&\s{B}\qw&\multigate{1}{\Braid_{1}}&\s{\prim{B}}\qw&\braidingGhost&\s{\secondE{A}}\qw&\multigate{1}{\Braid_{2}}&\s{D}\qw&\qw&
			\\
			&\pureghost{}&\pureghost{\Braid_{1}}&\qw&\s{\secondE{B}}\qw&\qw&\ghost{\Braid_{2}}&
		}, 
	\end{equation}
	where \system{A}, \system{B} are generic systems of the theory and \system{C}, \system{D} are systems such that \system{CD} is isomorphic to \system{AB}.
\end{theorem}

\begin{proof}
	Let us start by considering the ordered decomposition of \system{AB} in the set of subsystems on which $\Braid$ acts:
	\begin{equation*}
		\{\system{A}_{1},\dots,\system{A}_{n},\system{B}_{1},\dots,\system{B}_{m}\}.
	\end{equation*} 
	The action of $\Braid$ is to permute the order of these subsystems:
	\begin{equation*}
		\{\system{A}_{1},\dots,\system{A}_{n},\system{B}_{1},\dots,\system{B}_{m}\}  
	\end{equation*}
	\begin{equation*}
		\downarrow \Braid
	\end{equation*}
	\begin{equation*}
		\{\sigma\left( \system{A}_{1} \right), \dots, \sigma\left( \system{A}_{n} \right),\sigma\left( \system{B}_{1} \right),\dots,\sigma\left( \system{B}_{m} \right)\} = \{\system{C}_{1},\dots,\system{C}_{l},\system{D}_{1},\dots,\system{D}_{k}\}.
	\end{equation*}
	If we now define $N = \left\{ 1, \dots, n \right\}, M = \left\{ 1, \dots, m \right\}, L = \left\{ 1, \dots, l \right\}, K = \left\{ 1, \dots, k \right\}$, the most general transformation that can happen due to the action of $\Braid$ is that
	\begin{equation*}
		\begin{aligned}
			&\left\{ \system{A}_{i}\right\} _{i \in \prim{N}} \transfArrow{\Braid} \left\{ \system{C}_{i}\right\} _{i \in \prim{L}},\\
			&\left\{ \system{A}_{j}\right\} _{j \in \secondE{N}} \transfArrow{\Braid} \left\{ \system{D}_{j}\right\}_{j \in \prim{K}},
		\end{aligned}
	\end{equation*}
	where $N = \prim{N} \bigcup \secondE{N}, \#\prim{N} = \#\prim{L}, \#\secondE{N} = \#\prim{K}$ and analogously for \system{B},
	\begin{equation*}
		\begin{aligned}
			&\left\{ \system{B}_{i}\right\} _{i \in \prim{M}} \transfArrow{\Braid} \left\{ \system{C}_{i}\right\} _{i \in \secondE{L}},\\
			&\left\{ \system{B}_{j}\right\} _{j \in \secondE{M}} \transfArrow{\Braid} \left\{ \system{D}_{j}\right\}_{j \in \secondE{K}},
		\end{aligned}
	\end{equation*}
	where $M = \prim{M} \bigcup \secondE{M}, \#\prim{M} = \#\secondE{L}, \#\secondE{M} = \#\secondE{K}$, and overall $L = \prim{L} \bigcup \secondE{L}, K = \prim{K} \bigcup \secondE{K}$. With $\# S$ we denote the cardinality of the set $S$.\\
	Now we want to show that this permutation can always be achieved thorough a transformation with the same form as that of \eqref{eqt:OPT:permutation:symm:formula}.\\
	We begin by observing that in the case of system \system{A} one can always find a permutation that reorganizes the systems in such a way that the ones that are mapped into states of \system{C} are on the top and the ones that are mapped into \system{D} are on the bottom,
	\begin{equation*}
		\myQcircuitSupMat{
			&\s{A}\qw&\gate{\Braid_{3}}&\qw&\sEnsembleDouble{A}{i}{N'}{j}{N''}\qw&\qw&\splitter&\qw&\sEnsemble{A}{i}{N'}\qw&\qw&\qw&
			\\
			&\pureghost{}&\pureghost{\Braid_{3}}&\pureghost{}&\pureghost{}&\pureghost{}&\splitterGhost&\qw&\sEnsemble{A}{j}{N''}\qw&\qw&\qw&
		},
	\end{equation*}
	where the ordering of the $\system{A}_{i}$ and $\system{A}_{j}$ is not important. We can then suppose that the same happens also to the subsystems of \system{B},
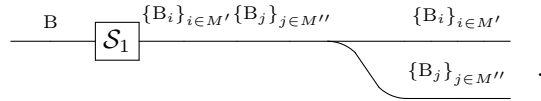
	\begin{equation*}
		\myQcircuitSupMat{
			&\s{B}\qw&\gate{\Braid_{1}}&\qw&\sEnsembleDouble{B}{i}{M'}{j}{M''}\qw&\qw&\splitter&\qw&\sEnsemble{B}{i}{M'}\qw&\qw&\qw&
			\\
			&\pureghost{}&\pureghost{\Braid_{1}}&\pureghost{}&\pureghost{}&\pureghost{}&\splitterGhost&\qw&\sEnsemble{B}{j}{M''}\qw&\qw&\qw&
		}.
	\end{equation*}
	Now we have to take the subsystems of \system{A} that are mapped into \system{D} and move them down, and vice versa for the ones of \system{B} that are mapped into \system{C}. This can be achieved by swapping $\left\{ \system{A}_{j}\right\}_{j \in \secondE{N}}$ with $\left\{ \system{B}_{i}\right\} _{i \in \prim{M}}$:
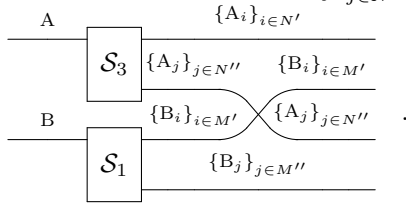
	\begin{equation*}
		\myQcircuitSupMat{
			&\s{A}\qw&\multigate{1}{\Braid_{3}}&\qw&\qw&\sEnsemble{A}{i}{N'}\qw&\qw&\qw&\qw&
			\\
			&\pureghost{}&\pureghost{\Braid_{3}}&\qw&\sEnsemble{A}{j}{N''}\qw&\braidingSym&\sEnsemble{B}{i}{M'}\qw&\qw&\qw&
			\\
			&\s{B}\qw&\multigate{1}{\Braid_{1}}&\qw&\sEnsemble{B}{i}{M'}\qw&\braidingGhost&\sEnsemble{A}{j}{N''}\qw&\qw&\qw&
			\\
			&\pureghost{}&\pureghost{\Braid_{1}}&\qw&\qw&\sEnsemble{B}{j}{M''}\qw&\qw&\qw&\qw&
		}.
	\end{equation*}
	Now to conclude we need only to add two permutations $\Braid^{2}, \Braid^{4}$, that can always be found, to correctly reorder the subsystems to obtain \system{C} and \system{D}:
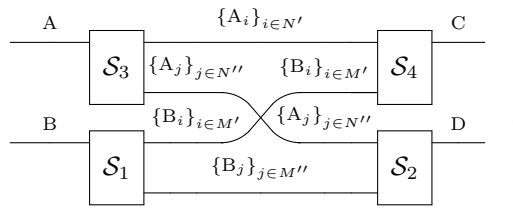
	\begin{equation*}
		\myQcircuitSupMat{
			&\s{A}\qw&\multigate{1}{\Braid_{3}}&\qw&\qw&\sEnsemble{A}{i}{N'}\qw&\qw&\qw&\multigate{1}{\Braid_{4}}&\s{C}\qw&\qw&
			\\
			&\pureghost{}&\pureghost{\Braid_{3}}&\qw&\sEnsemble{A}{j}{N''}\qw&\braidingSym&\sEnsemble{B}{i}{M'}\qw&\qw&\ghost{\Braid_{4}}\qw&
			\\
			&\s{B}\qw&\multigate{1}{\Braid_{1}}&\qw&\sEnsemble{B}{i}{M'}\qw&\braidingGhost&\sEnsemble{A}{j}{N''}\qw&\qw&\multigate{1}{\Braid_{2}}&\s{D}\qw&\qw&
			\\
			&\pureghost{}&\pureghost{\Braid_{1}}&\qw&\qw&\sEnsemble{B}{j}{M''}\qw&\qw&\qw&\ghost{\Braid_{2}}&
		}.
	\end{equation*}
	
	Therefore we have shown that, for any permutation $\Braid$, it is always possible to find a transformation such as the one in \eqref{eqt:OPT:permutation:symm:formula} that permutes the systems as $\Braid$. From the fact that permutations can be completely characterized by how they permute its input systems, then the equality between the two transformations follows.\\
	
\end{proof}

\begin{remark}
	\label{obs:OPT:permutation:generalDecomposition:empty}
	We highlight that, in general, in the preceding theorem \system{A}, \system{B}, \system{C}, \system{D} can be the trivial system and this holds also for \system{\prim{A}}, \system{\secondE{A}}, \system{\prim{B}}, \system{\secondE{B}}.
\end{remark}

\section{Minimal operational probabilistic theories}\label{app:mopt}

\begin{definition}[Minimal operational probabilistic theory (MOPT)]
	\label{def:OPT:minimal}
	We define as MOPT an OPT where the only allowed tests are the ones obtainable by composing the elements of
	\begin{equation}
		\label{eqt:OPT:minimal:def:test}
		\left\{ 
		\myQcircuit{
			&\s{A}\qw&\qw&
		}
		\quad , \quad
		\myQcircuit{
			&\s{A}\qw&\braidingSym&\s{B}\qw&\qw&
			\\
			&\s{B}\qw&\braidingGhost&\s{A}\qw&\qw&
		}
		\quad , \quad
		\myQcircuit{
			&\prepareC{\eventPreparationTestNoKet{\rho}{i}{I}}&\s{A}\qw&\qw&
		}
		\quad , \quad
		\myQcircuit{
			&\s{A}\qw&\measureD{\eventObservationTestNoBra{a}{j}{J}}&
		}
		\right\},
	\end{equation} 
	where $\eventPreparationTest{\rho}{i}{I}$ and $\eventObservationTest{a}{j}{J}$ are all the possible preparation test and observation tests of the theory, and the limits of all the Cauchy sequences of tests of this type. Thus the only allowed transformations are those obtainable by sequential and parallel composition of the elements of
	\begin{equation}
		\label{eqt:OPT:minimal:def:event}
		\left\{ 
		\myQcircuit{
			&\s{A}\qw&\qw&
		}
		\quad , \quad
		\myQcircuit{
			&\s{A}\qw&\braidingSym&\s{B}\qw&\qw&
			\\
			&\s{B}\qw&\braidingGhost&\s{A}\qw&\qw&
		}
		\quad , \quad
		\myQcircuit{
			&\prepareC{\rho}&\s{A}\qw&\qw&
		}
		\quad , \quad
		\myQcircuit{
			&\s{A}\qw&\measureD{a}&
		}
		\right\},
	\end{equation}
	for every $\system{A}, \system{B} \in \Sys{\OPT}$, $\rket{\rho} \in \St{A}$, and $\rbra{a} \in \Eff{A}$, and the limits of all the Cauchy sequences of transformations of this type that belong to a test of the theory.
\end{definition}
We observe that these are the minimum requirements that can be made on an OPT to cope with the required compositional structure and the Cauchy completeness. In other words, if any of the elements of \eqref{eqt:OPT:minimal:def:test}, or equivalently \eqref{eqt:OPT:minimal:def:event}, or of the limits were removed, the theory could no longer be classified as an OPT.


\begin{theorem}
	\label{thm:OPT:minimal:symmetric:typeD:generalTransf}
	In every MOPT any transformation $\eventNoDown{T} \in \Transf{A}{B}$ obtained as a parallel and sequential composition of the elements of \eqref{eqt:OPT:minimal:def:event} is of the form
	\begin{equation}
		\label{eqt:OPT:minimal:symmetric:transf:generic}
		\myQcircuit{
			&\s{A}\qw&\gate{\eventNoDown{T}}&\s{B}\qw&\qw&
		} =
		\myQcircuit{
			&\pureghost{}&\multiprepareC{1}{\preparationEventNoDown{\rho}}&\qw&\s{C}\qw&\qw&\multimeasureD{1}{\observationEventNoDown{a}}&
			\\
			&\pureghost{}&\pureghost{\preparationEventNoDown{\rho}}&\s{\prim{B}}\qw&\braidingSym&\s{\prim{A}}\qw&\ghost{\observationEventNoDown{a}}&
			\\
			&\s{A}\qw&\multigate{1}{\Braid_{1}}&\s{\prim{A}}\qw&\braidingGhost&\s{\prim{B}}\qw&\multigate{1}{\Braid_{2}}&\s{B}\qw&\qw&
			\\
			&\pureghost{}&\pureghost{\Braid_{1}}&\qw&\s{E}\qw&\qw&\ghost{\Braid_{2}}&
		},
	\end{equation}
	where $\Braid_{1}, \Braid_{2} \in \RevTransfA{\OPT}$ are appropriate permutations, $\rketSystem{\preparationEventNoDown{\rho}}{C \prim{B}} \in \St{C \prim{B}}$, $\rbraSystem{\observationEventNoDown{a}}{C\prim{A}} \in \Eff{C \prim{A}}$, and \system{A}, \system{B}, $\system{\prim{A}}$, $\system{\prim{B}}$, $\system{C}$, $\system{E} \in \Sys{\OPT}$ may also be equal to the trivial system.
\end{theorem}
\begin{proof}
	To prove that this result, we will start by showing that every transformation can be written in the form 
	\begin{equation}
		\label{eqt:OPT:minimal:symmetric:transf:part}
		\myQcircuit{
			&\s{A}\qw&\gate{\eventNoDown{T}}&\s{B}\qw&\qw&
		} = 
		\myQcircuit{
			&\prepareC{\prim{\preparationEventNoDown{\rho}}}&\s{\prim{C}}\qw&\multigate{1}{\Braid}&\s{\prim{D}}\qw&\measureD{\prim{\observationEventNoDown{a}}}&
			\\
			&\s{A}\qw&\qw&\ghost{\Braid}&\qw&\s{B}\qw&\qw&
		}, 
	\end{equation}
	Let us consider the decomposition of $\eventNoDown{T}$ in its constituent elements, \eqref{eqt:OPT:minimal:def:event}, and focus our attention on one of the measurements in it. An effect was chosen, but the procedure remains the same even if one chooses to start with a state. In the case in which neither of them are included in the decomposition it means that $\eventNoDown{T} = \Braid$, i.e., \eqref{eqt:OPT:minimal:symmetric:transf:part} with $\system{C}' = \system{D}' = \trivialSystem$.\\
	In the case in which an effect $\rbra{\observationEventNoDown{a}_{1}}$ is present, it is possible to isolate it and rewrite the transformation in the following way,
	\begin{equation*}
		\myQcircuitSupMat{
			&\s{A}\qw&\gate{\eventNoDown{T}}&\s{B}\qw&\qw&
		} =
		\myQcircuitSupMat{
			&\s{A}\qw&\multigate{2}{\eventNoDown{T}_{1}}&\qw&\s{D_{2}}\qw&\qw&\multigate{2}{\eventNoDown{T}_{2}}&\s{B}\qw&\qw&
			\\
			&\pureghost{}&\pureghost{\eventNoDown{T}_{1}}&\s{D_{1}}\qw&\measureD{\observationEventNoDown{a}_{1}}&\pureghost{}&\pureghost{\eventNoDown{T}_{2}}&\pureghost{}&
			\\
			&\pureghost{}&\pureghost{\eventNoDown{T}_{1}}&\qw&\s{D_{3}}\qw&\qw&\ghost{\eventNoDown{T}_{2}}&\pureghost{}
		},
	\end{equation*}
	where $\eventNoDown{T}_{1}$ and $\eventNoDown{T}_{2}$ are such that $\eventNoDown{T} = \eventNoDown{T}_{2} \seqC \left( \identityTest{D_{2}} \paralC \rbra{\observationEventNoDown{a}_{1}} \paralC \identityTest{D_{3}} \right) \seqC \eventNoDown{T}_{1}$ and $\system{D}_{1}, \system{D_{2}}, \system{D_{3}} \in \Sys{\OPT}$ are appropriate systems. It is not excluded the possibility of $\system{D_{2}}, \system{D_{3}}$ being the trivial system.\\
	Using the reversibility of the permutations, it is possible to write
	\allowDisplayBreaks
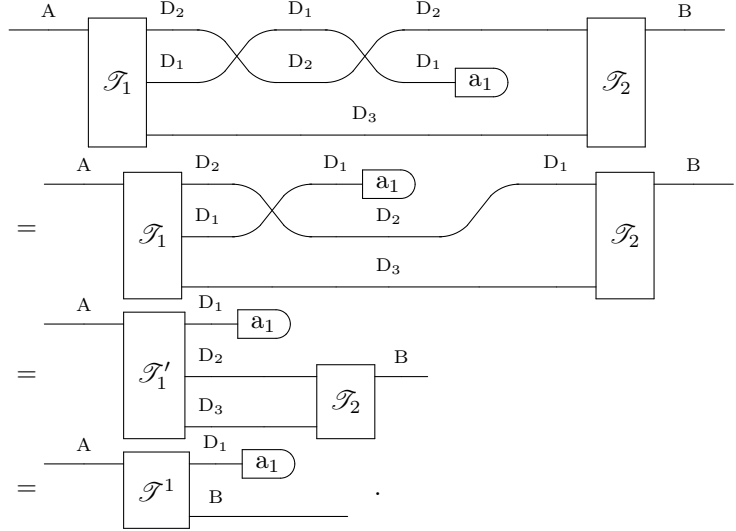
{
		\begin{equation*}
			\begin{aligned}
				& \myQcircuitSupMat{
					&\s{A}\qw&\multigate{2}{\eventNoDown{T}_{1}}&\s{D_{2}}\qw&\braidingSym&\s{D_{1}}\qw&\braidingSym&\s{D_{2}}\qw&\qw&\qw&\multigate{2}{\eventNoDown{T}_{2}}&\s{B}\qw&\qw&
					\\
					&\pureghost{}&\pureghost{\eventNoDown{T}_{1}}&\s{D_{1}}\qw&\braidingGhost&\s{D_{2}}\qw&\braidingGhost&\s{D_{1}}\qw&\measureD{\observationEventNoDown{a}_{1}}&\pureghost{}&\pureghost{\eventNoDown{T}_{2}}&\pureghost{}&
					\\
					&\pureghost{}&\pureghost{\eventNoDown{T}_{1}}&\qw&\qw&\qw&\s{D_{3}}\qw&\qw&\qw&\qw&\ghost{\eventNoDown{T}_{2}}&\pureghost{}
				}  \\
				&  = \myQcircuitSupMat{
					&\s{A}\qw&\multigate{2}{\eventNoDown{T}_{1}}&\s{D_{2}}\qw&\braidingSym&\s{D_{1}}\qw&\measureD{\observationEventNoDown{a}_{1}}&\braidingInvId&\s{D_{1}}\qw&\multigate{2}{\eventNoDown{T}_{2}}&\s{B}\qw&\qw&
					\\
					&\pureghost{}&\pureghost{\eventNoDown{T}_{1}}&\s{D_{1}}\qw&\braidingGhost&\qw&\s{D_{2}}\qw&\braidingGhost&\pureghost{}&\pureghost{\eventNoDown{T}_{2}}&\pureghost{}&
					\\
					&\pureghost{}&\pureghost{\eventNoDown{T}_{1}}&\qw&\qw&\qw&\s{D_{3}}\qw&\qw&\qw&\ghost{\eventNoDown{T}_{2}}&\pureghost{}
				}  \\
				& = \myQcircuitSupMat{
					&\s{A}\qw&\multigate{2}{\eventNoDown{T}'_{1}}&\s{D_{1}}\qw&\measureD{\observationEventNoDown{a}_{1}}&
					\\
					&\pureghost{}&\pureghost{\eventNoDown{T}'_{1}}&\s{D_{2}}\qw&\qw&\multigate{1}{\eventNoDown{T}_{2}}&\s{B}\qw&\qw&
					\\
					&\pureghost{}&\pureghost{\eventNoDown{T}'_{1}}&\s{D_{3}}\qw&\qw&\ghost{\eventNoDown{T}_{2}}&
				}  \\
				& = \myQcircuitSupMat{
					&\s{A}\qw&\multigate{1}{\eventNoDown{T}^{1}}&\s{D_{1}}\qw&\measureD{\observationEventNoDown{a}_{1}}&
					\\
					&\pureghost{}&\pureghost{\eventNoDown{T}^{1}}&\s{B}\qw&\qw&\qw&\qw&
				}.
			\end{aligned}
		\end{equation*}
	}
	
	Now it is sufficient to iterate the procedure on $\eventNoDown{T}^{1}$ until, after $n$ steps, one obtains a transformation $\eventNoDown{T}^{n} = \Braid$. The result would be something of the form
	\begin{equation*}
		\myQcircuitSupMat{
			&\prepareC{\preparationEventNoDown{\rho}_{1}}&\s{C_1}\qw&\multigate{3}{\eventNoDown{T}^{n}}&\s{D_1}\qw&\measureD{\observationEventNoDown{a}_{1}}&
			\\
			&\vdots&\vdots&\pureghost{\eventNoDown{T}^{n}}&\vdots&\vdots&
			\\
			&\prepareC{\preparationEventNoDown{\rho}_{l}}&\s{C_l}\qw&\ghost{\eventNoDown{T}^{n}}&\s{D_m}\qw&\measureD{\observationEventNoDown{a}_{m}}&
			\\
			&\s{A}\qw&\qw&\ghost{\eventNoDown{T}^{n}}&\qw&\s{B}\qw&\qw&
		} =
		\myQcircuitSupMat{
			&\prepareC{\rho}&\s{C'}\qw&\multigate{1}{\Braid}&\s{D'}\qw&\measureD{a}&
			\\
			&\s{A}\qw&\qw&\ghost{\Braid}&\qw&\s{B}\qw&\qw&
		}.
	\end{equation*} 
	
	We can now apply \autoref{thm:OPT:symmetric:permutations:generalForm} to \eqref{eqt:OPT:minimal:symmetric:transf:part} obtaining
	\begin{equation*}
		\myQcircuit{
			&\s{A}\qw&\gate{\eventNoDown{T}}&\s{B}\qw&\qw&
		} =
		\myQcircuit{
			&\prepareC{\prim{\preparationEventNoDown{\rho}}}&\s{\prim{C}}\qw&\multigate{1}{\Braid_{3}}&\qw&\s{C}\qw&\qw&\multigate{1}{\Braid_{4}}&\s{\prim{D}}\qw&\measureD{\prim{\observationEventNoDown{a}}}&
			\\
			&\pureghost{}&\pureghost{}&\pureghost{\Braid_{3}}&\s{\prim{B}}\qw&\braidingSym&\s{\prim{A}}\qw&\ghost{\Braid_{4}}&
			\\
			&\pureghost{}&\s{A}\qw&\multigate{1}{\Braid_{1}}&\s{\prim{A}}\qw&\braidingGhost&\s{\prim{B}}\qw&\multigate{1}{\Braid_{2}}&\s{B}\qw&\qw&
			\\
			&\pureghost{}&\pureghost{}&\pureghost{\Braid_{1}}&\qw&\s{E}\qw&\qw&\ghost{\Braid_{2}}&
		}.
	\end{equation*}
	Now absorbing $\Braid_{3}, \Braid_{4}$ into $\rketSystem{\prim{\preparationEventNoDown{\rho}}}{\prim{C}}$ and $\rbraSystem{\prim{\observationEventNoDown{a}}}{\prim{D}}$, respectively, the proof is concluded:
	\begin{equation*}
		\myQcircuit{
			&\s{A}\qw&\gate{\eventNoDown{T}}&\s{B}\qw&\qw&
		} =
		\myQcircuit{
			&\pureghost{}&\multiprepareC{1}{\preparationEventNoDown{\rho}}&\qw&\s{C}\qw&\qw&\multimeasureD{1}{\observationEventNoDown{a}}&
			\\
			&\pureghost{}&\pureghost{\preparationEventNoDown{\rho}}&\s{\prim{B}}\qw&\braidingSym&\s{\prim{A}}\qw&\ghost{\observationEventNoDown{a}}&
			\\
			&\s{A}\qw&\multigate{1}{\Braid_{1}}&\s{\prim{A}}\qw&\braidingGhost&\s{\prim{B}}\qw&\multigate{1}{\Braid_{2}}&\s{B}\qw&\qw&
			\\
			&\pureghost{}&\pureghost{\Braid_{1}}&\qw&\s{E}\qw&\qw&\ghost{\Braid_{2}}&
		}.
	\end{equation*}
\end{proof}


One important check for a well-defined OPT is that the spaces of transformations must be closed under parallel and sequential composition. This can be easily proved to hold in every MOPT by exploiting \eqref{eqt:OPT:minimal:symmetric:transf:part}.\\
Let us start by demonstrating the case of sequential composition:
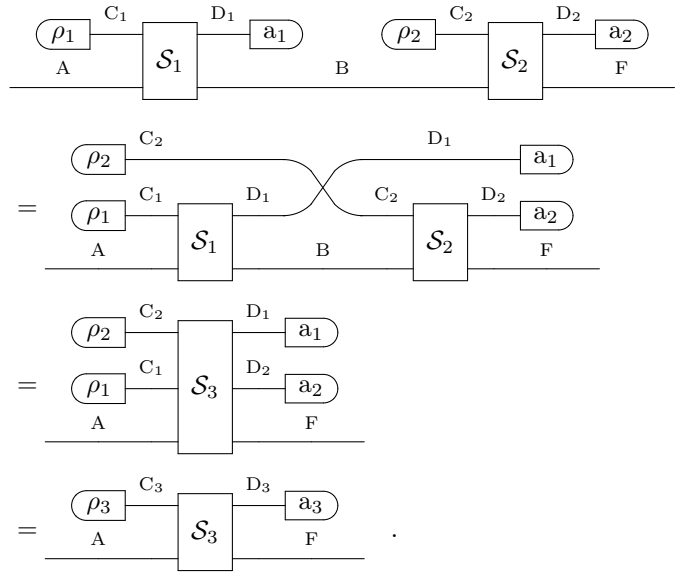
\allowdisplaybreaks{
	\begin{align*}
		\label{eqt:OPT:minimal:transf:sequential}
		\begin{split}
			& \myQcircuitSupMat{
				&\prepareC{\preparationEventNoDown{\rho}_{1}}&\s{C_{1}}\qw&\multigate{1}{\Braid_{1}}&\s{D_{1}}\qw&\measureD{\observationEventNoDown{a}_{1}}&\pureghost{}&\prepareC{\preparationEventNoDown{\rho}_{2}}&\s{C_{2}}\qw&\multigate{1}{\Braid_{2}}&\s{D_{2}}\qw&\measureD{\observationEventNoDown{a}_{2}}&
				\\
				&\s{A}\qw&\qw&\ghost{\Braid_{1}}&\qw&\qw&\s{B}\qw&\qw&\qw&\ghost{\Braid_{2}}&\qw&\s{F}\qw&\qw&
			}  \\[10pt]
			& = \myQcircuitSupMat{
				&\prepareC{\preparationEventNoDown{\rho}_{2}}&\s{C_{2}}\qw&\qw&\qw&\braidingSym&\qw&\s{D_{1}}\qw&\qw&\measureD{\observationEventNoDown{a}_{1}}&
				\\
				&\prepareC{\preparationEventNoDown{\rho}_{1}}&\s{C_{1}}\qw&\multigate{1}{\Braid_{1}}&\s{D_{1}}\qw&\braidingGhost&\s{C_{2}}\qw&\multigate{1}{\Braid_{2}}&\s{D_{2}}\qw&\measureD{\observationEventNoDown{a}_{2}}&
				\\
				&\s{A}\qw&\qw&\ghost{\Braid_{1}}&\qw&\s{B}\qw&\qw&\ghost{\Braid_{2}}&\qw&\s{F}\qw&\qw&
			}  \\[10pt]
			& = \myQcircuitSupMat{
				&\prepareC{\preparationEventNoDown{\rho}_{2}}&\s{C_{2}}\qw&\multigate{2}{\Braid_{3}}&\s{D_{1}}\qw&\measureD{\observationEventNoDown{a}_{1}}&
				\\
				&\prepareC{\preparationEventNoDown{\rho}_{1}}&\s{C_{1}}\qw&\ghost{\Braid_{3}}&\s{D_{2}}\qw&\measureD{\observationEventNoDown{a}_{2}}&
				\\
				&\s{A}\qw&\qw&\ghost{\Braid_{3}}&\qw&\s{F}\qw&\qw&
			}  \\[10pt]
			& = \myQcircuitSupMat{
				&\prepareC{\preparationEventNoDown{\rho}_{3}}&\s{C_{3}}\qw&\multigate{1}{\Braid_{3}}&\s{D_{3}}\qw&\measureD{\observationEventNoDown{a}_{3}}&
				\\
				&\s{A}\qw&\qw&\ghost{\Braid_{3}}&\qw&\s{F}\qw&\qw&
			}.
		\end{split}
	\end{align*}
}
The proof for parallel composition is analogous,
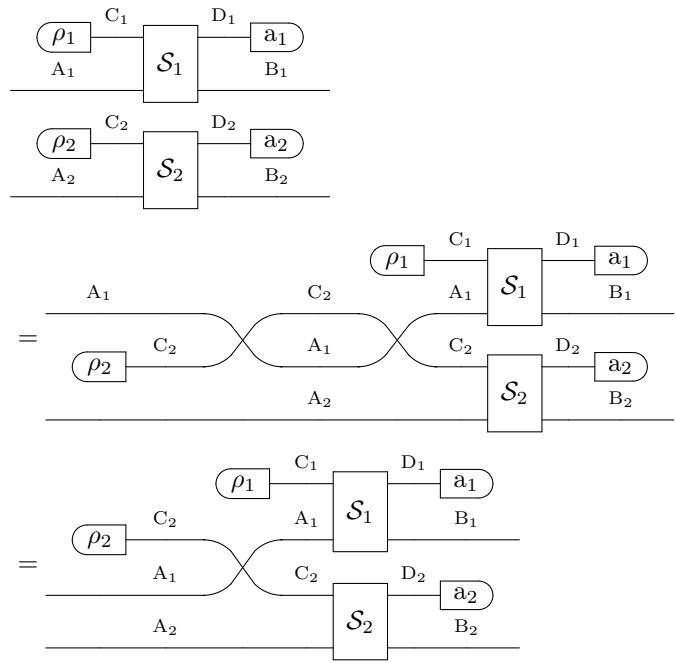
\allowDisplayBreaks{
	\begin{align*}
		\begin{split}
			& \myQcircuitSupMat{
				&\prepareC{\preparationEventNoDown{\rho}_{1}}&\s{C_{1}}\qw&\multigate{1}{\Braid_{1}}&\s{D_{1}}\qw&\measureD{\observationEventNoDown{a}_{1}}&
				\\
				&\s{A_{1}}\qw&\qw&\ghost{\Braid_{1}}&\qw&\s{B_{1}}\qw&\qw&
				\\
				&\prepareC{\preparationEventNoDown{\rho}_{2}}&\s{C_{2}}\qw&\multigate{1}{\Braid_{2}}&\s{D_{2}}\qw&\measureD{\observationEventNoDown{a}_{2}}&
				\\
				&\s{A_{2}}\qw&\qw&\ghost{\Braid_{1}}&\qw&\s{B_{2}}\qw&\qw&
			}  \\[10pt]
			& = \myQcircuitSupMat{
				&\pureghost{}&\pureghost{}&\pureghost{}&\pureghost{}&\prepareC{\preparationEventNoDown{\rho}_{1}}&\s{C_{1}}\qw&\multigate{1}{\Braid_{1}}&\s{D_{1}}\qw&\measureD{\observationEventNoDown{a}_{1}}&
				\\
				&\s{A_{1}}\qw&\qw&\braidingSym&\s{C_{2}}\qw&\braidingSym&\s{A_{1}}\qw&\ghost{\Braid_{1}}&\qw&\s{B_{1}}\qw&\qw&
				\\
				&\prepareC{\preparationEventNoDown{\rho}_{2}}&\s{C_{2}}\qw&\braidingGhost&\s{A_{1}}\qw&\braidingGhost&\s{C_{2}}\qw&\multigate{1}{\Braid_{2}}&\s{D_{2}}\qw&\measureD{\observationEventNoDown{a}_{2}}&
				\\
				&\qw&\qw&\qw&\s{A_{2}}\qw&\qw&\qw&\ghost{\Braid_{1}}&\qw&\s{B_{2}}\qw&\qw&
			}  \\[10pt]
			& = \myQcircuitSupMat{
				&\pureghost{}&\pureghost{}&\prepareC{\preparationEventNoDown{\rho}_{1}}&\s{C_{1}}\qw&\multigate{1}{\Braid_{1}}&\s{D_{1}}\qw&\measureD{\observationEventNoDown{a}_{1}}&
				\\
				&\prepareC{\preparationEventNoDown{\rho}_{2}}&\s{C_{2}}\qw&\braidingSym&\s{A_{1}}\qw&\ghost{\Braid_{1}}&\qw&\s{B_{1}}\qw&\qw&
				\\
				&\qw&\s{A_{1}}\qw&\braidingGhost&\s{C_{2}}\qw&\multigate{1}{\Braid_{2}}&\s{D_{2}}\qw&\measureD{\observationEventNoDown{a}_{2}}&
				\\
				&\qw&\s{A_{2}}\qw&\qw&\qw&\ghost{\Braid_{1}}&\qw&\s{B_{2}}\qw&\qw&
			}
		\end{split}
	\end{align*}
}
\noindent{}and applying the same procedure on the right-hand side of the circuit with $\rbra{\observationEventNoDown{a}_{2}}$ one obtains
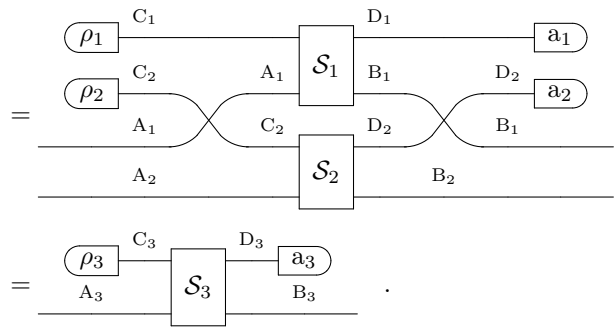
\allowDisplayBreaks{
	\begin{align*}
		& = \myQcircuitSupMat{
			&\prepareC{\preparationEventNoDown{\rho}_{1}}&\s{C_{1}}\qw&\qw&\qw&\multigate{1}{\Braid_{1}}&\s{D_{1}}\qw&\qw&\qw&\measureD{\observationEventNoDown{a}_{1}}&
			\\
			&\prepareC{\preparationEventNoDown{\rho}_{2}}&\s{C_{2}}\qw&\braidingSym&\s{A_{1}}\qw&\ghost{\Braid_{1}}&\s{B_{1}}\qw&\braidingSym&\s{D_{2}}\qw&\measureD{\observationEventNoDown{a}_{2}}&
			\\
			&\qw&\s{A_{1}}\qw&\braidingGhost&\s{C_{2}}\qw&\multigate{1}{\Braid_{2}}&\s{D_{2}}\qw&\braidingGhost&\s{B_{1}}\qw&\qw&\qw&
			\\
			&\qw&\s{A_{2}}\qw&\qw&\qw&\ghost{\Braid_{2}}&\qw&\s{B_{2}}\qw&\qw&\qw&\qw&
		}  \\[10pt]
		& = \myQcircuitSupMat{
			&\prepareC{\preparationEventNoDown{\rho}_{3}}&\s{C_{3}}\qw&\multigate{1}{\Braid_{3}}&\s{D_{3}}\qw&\measureD{\observationEventNoDown{a}_{3}}&
			\\
			&\s{A_{3}}\qw&\qw&\ghost{\Braid_{3}}&\qw&\s{B_{3}}\qw&\qw&
		}.
	\end{align*}
}
\subsection{Properties of MOPTs with the causality assumption}
We conclude this Appendix by proving the stability result for the form of the deterministic transformations of MOPTs.
\begin{lemma}
	\label{lem:OPT:minimal:transf:causalDeterm}
	In a causal MOPT every deterministic transformation obtained as composition of the elements in \eqref{eqt:OPT:minimal:def:event} is of the form 
	\begin{equation}
		\label{eqt:OPT:minimal:transf:lem:causalDeterm}
		\minimalDeterministicCausalDestroyReprep{A}{B}{\prim{A}}{\prim{B}}{E}{\Braid^{(1)}}{\Braid^{(2)}}{\rho}.
	\end{equation}
\end{lemma}

\begin{remark}
	The transformation
	\begin{equation*}
		\measurePrepare{\prim{A}}{\prim{B}}{\rho}{\observationUniqueDeterministic}
	\end{equation*}
	between the two braid transformations in \eqref{eqt:OPT:minimal:transf:lem:causalDeterm}
	is sometimes referred to as ``destroy and reprepare,'' since whatever the input it will ``destroy'' it and prepare the state $\rket{\rho}$.
\end{remark}

\begin{theorem}
	\label{thm:OPT:minimal:transf:goodDeterministic}
	In a causal MOPT the limits of Cauchy sequences of deterministic transformations are still of the form \eqref{eqt:OPT:minimal:transf:lem:causalDeterm}.
\end{theorem}
\begin{proof}
	Let us start by considering a Cauchy sequence of deterministic transformations from \system{A} to \system{B}, which by \eqref{eqt:OPT:minimal:transf:lem:causalDeterm} we know to be of the form
	\begin{equation}
		\label{proof:OPT:minimal:transf:thm:goodDeterministic:1}
		\left\{
		\minimalDeterministicCausalDestroyReprepSequencePrime{A}{B}{A}{B}{E}{\Braid^{(1)}_{n}}{\Braid^{(2)}_{n}}{\rho_{n}}{n}
		\right\}_{n \in \mathbb{N}}.
	\end{equation}
	
	The proof can now be subdivided into three steps:
	\begin{itemize}
		\item Given that the two systems \system{A} and \system{B} can only be a composition of a finite number of systems, the sets of permutations that have this systems respectively as input and output [\permutationCollectionAB{A}{E} and $\permutationCollectionAB{\prim{E}}{B}$ for all appropriate systems \system{E},$\system{\prim{E}} \in \Sys{\OPT}$] are finite.

		Consequently, due to the fact that we have a sequence, i.e., infinite terms, there must exist at least a couple of permutations $\Braid^{1}$ and $\Braid^{2}$ that appear infinitely many times together ``on the outside'' of the elements of the sequence \eqref{proof:OPT:minimal:transf:thm:goodDeterministic:1}. We can now concentrate on the subsequence with this couple of permutations
		\begin{equation*}
			\left\{ \minimalDeterministicCausalDestroyReprepSequencePrime{A}{B}{A}{B}{E}{\Braid^{(1)}}{\Braid^{(2)}}{\rho_{n}}{n} \right\}_{n \in \mathbb{N}}.
		\end{equation*}
		Since \eqref{proof:OPT:minimal:transf:thm:goodDeterministic:1} is a Cauchy sequence, also its subsequences will be Cauchy and they will have the same limit.
		
		\item We now focus our attention on the systems $\system{E}_{n}$. Due to the fact that in the previous point we have fixed $\Braid^{1}$, the systems contained within the composite system $\system{A}'_{n}\system{E}_{n}$ will not change. Therefore, the only change that can occur at the variation of $n$ is how they are grouped. \\
		For example, if $\system{A}'_{n} = \system{S_{1}}$ and $\system{E}_{n} = \left(\system{S_{2}} \system{S}_{3} \system{S_{4}}\right)$, for a different value $\prim{n} \neq n$, it must be  $\system{A}'_{n'} = \system{S_{1}}\system{S_{2}}$ and $\system{E}_{n'} = \system{S_{3}} \system{S_{4}}$, or $\system{A}'_{n'} = \left(\system{S_{1}}\system{S_{2}}\system{S_{3}}\right)$ and $\system{E}_{n'} = \system{S_{4}}$, or any other possible regrouping (also the original one) in which the order of the $\system{S_{i}}$ does not change.\\
		Given that  $\system{A}'_{n}\system{E}_{n}$ can be composed only of a finite number of systems, and analogously for $\system{B}'_{n}\system{E}_{n}$, it is always possible to find at least a system \system{E} that appears infinitely many times in the considered sequence. By fixing \system{E}, then also the systems $\system{A}'$ and $\system{B}'$ are automatically fixed. Proceeding exactly as in the previous point we will focus from now on the subsequence where these systems are fixed:
		\begin{equation*}
			\left\{ \minimalDeterministicCausalDestroyReprep{A}{B}{\prim{A}}{\prim{B}}{E}{\Braid^{(1)}}{\Braid^{(2)}}{\rho_{n}} \right\}_{n \in \mathbb{N}}.
		\end{equation*}	
		
		\item Considering this subsequence we can now easily see that the following relation holds $\forall n,m \in \mathbb{N}$,
		\allowDisplayBreaks{
			\begin{align*}
				& \left\lVert \quad \minimalDeterministicCausalDestroyReprep{A}{B}{\prim{A}}{\prim{B}}{E}{\Braid^{(1)}}{\Braid^{(2)}}{\rho_{n}}  \right. \\[10pt] & \hspace{1cm} - \; \left. \minimalDeterministicCausalDestroyReprep{A}{B}{\prim{A}}{\prim{B}}{E}{\Braid^{(1)}}{\Braid^{(2)}}{\rho_{m}}
				\right\rVert_{op}  \\[10pt] & = 
				\normOp{ \quad
					\myQcircuit{
						&\s{\prim{A}}\qw&\measureD{\observationUniqueDeterministic}&\pureghost{}&\prepareC{\eventNoDown{\rho}_{n}}&\s{\prim{B}}\qw&\qw&
						\\
						&\qw&\qw&\s{E}\qw&\qw&\qw&\qw&
					} - \quad
					\myQcircuit{
						&\s{\prim{A}}\qw&\measureD{\observationUniqueDeterministic}&\pureghost{}&\prepareC{\eventNoDown{\rho}_{m}}&\s{\prim{B}}\qw&\qw&
						\\
						&\qw&\qw&\s{E}\qw&\qw&\qw&\qw&
					}
				}  \\[10pt] & =
				\normOp{ \quad
					\myQcircuit{
						&\s{\prim{A}}\qw&\measureD{\observationUniqueDeterministic}&\prepareC{\eventNoDown{\rho}_{n}}&\s{\prim{B}}\qw&\qw&
					} - \quad
					\myQcircuit{
						&\s{\prim{A}}\qw&\measureD{\observationUniqueDeterministic}&\prepareC{\eventNoDown{\rho}_{m}}&\s{\prim{B}}\qw&\qw&
					}
				}  \\[10pt] & \geq
				\normOp{
					\myQcircuit{
						&\prepareC{\eventNoDown{\rho}_{n}}&\s{\prim{B}}\qw&\qw&
					} - \quad
					\myQcircuit{
						&\prepareC{\eventNoDown{\rho}_{m}}&\s{\prim{B}}\qw&\qw&
					}	
				},
			\end{align*}
		}
		where the norm used above is the \textdef{operational norm}~\cite{book:quantumTheoryFromFirstPrinciples} which has a nice operational interpretation: The distance between two transformations is related to the probability of discriminating them through the best possible procedure one can implement. This norm is well defined over the spaces of transformation since, as observed in the main text, these can be embedded in a real vector space. Furthermore, it satisfies the monotonicity property~\cite{book:quantumTheoryFromFirstPrinciples}
		\begin{equation*}
			\normOp{\eventNoDown{T}} \geq \normOp{\eventNoDown{E}\eventNoDown{T}\eventNoDown{C}},
		\end{equation*}	
		where $\eventNoDown{E} \in \Transf{C}{D}$ and $\eventNoDown{C} \in \Transf{A}{B}$ are deterministic transformations---the equality holds if both $\eventNoDown{E}$ and $\eventNoDown{C}$ are reversible---which is what was used in the last steps of the proof.\\
		
		What this implies is that the sequence of deterministic states of this particular subsequence of \eqref{proof:OPT:minimal:transf:thm:goodDeterministic:1} is Cauchy. 
		
		We can therefore conclude that the subsequence considered in this point, and consequently \eqref{proof:OPT:minimal:transf:thm:goodDeterministic:1}, converges to 
		\begin{equation*}
			\minimalDeterministicCausalDestroyReprep{A}{B}{\prim{A}}{\prim{B}}{E}{\Braid^{(1)}}{\Braid^{(2)}}{\rho},
		\end{equation*}
		where $\rketSystem{\eventNoDown{\rho}}{\prim{B}} = \lim_{n \to \infty} \rketSystem{\eventNoDown{\rho}_{n}}{\prim{B}}$. 
		With this we conclude our proof, since we found the desired result.
	\end{itemize}
\end{proof}

\begin{theorem}
	\label{thm:OPT:minimal:transf:stabilization}
	In a causal MOPT whenever one considers a Cauchy sequence of generic transformations obtained as parallel and sequential composition of the elements in \eqref{eqt:OPT:minimal:def:event},
	\begin{equation}
		\label{eqt:sequenceJelly}
		\left\{
			\myQcircuit{
				&\pureghost{}&\multiprepareC{1}{\preparationEventNoDown{\rho}_{n}}&\qw&\sSequence{C}{n}\qw&\qw&\multimeasureD{1}{\observationEventNoDown{a}_{n}}&
				\\
				&\pureghost{}&\pureghost{\preparationEventNoDown{\rho}_{n}}&\sSequencePrime{B}{n}\qw&\braidingSym&\sSequencePrime{A}{n}\qw&\ghost{\observationEventNoDown{a}_{n}}&
				\\				&\s{A}\qw&\multigate{1}{\Braid^{\left(1\right)}_{n}}&\sSequencePrime{A}{n}\qw&\braidingGhost&\sSequencePrime{B}{n}\qw&\multigate{1}{\Braid^{\left(2\right)}_{n}}&\s{B}\qw&\qw&
				\\
				&\pureghost{}&\pureghost{\Braid^{\left(1\right)}_{n}}&\qw&\sSequence{E}{n}\qw&\qw&\ghost{\Braid^{\left(2\right)}_{n}}&
			}
		\right\}_{n \in \mathbb{N}},
	\end{equation}	
	there always exists a subsequence where the systems $\system{E}_{n}$, $\system{\prim{A}}_{n}$, $\system{\prim{B}}_{n}$ and the permutations $\Braid^{\left(1\right)}_{n}$, $\Braid^{\left(2\right)}_{n}$ are fixed:
	\begin{equation*}
		\left\{
			\myQcircuit{
				&\pureghost{}&\multiprepareC{1}{\preparationEventNoDown{\rho}_{n}}&\qw&\sSequence{C}{n}\qw&\qw&\multimeasureD{1}{\observationEventNoDown{a}_{n}}&
				\\
				&\pureghost{}&\pureghost{\preparationEventNoDown{\rho}_{n}}&\s{B'}\qw&\braidingSym&\s{A'}\qw&\ghost{\observationEventNoDown{a}_{n}}&
				\\
				&\s{A}\qw&\multigate{1}{\Braid^{\left(1\right)}}&\s{A'}\qw&\braidingGhost&\s{B'}\qw&\multigate{1}{\Braid^{\left(2\right)}}&\s{B}\qw&\qw&
				\\
				&\pureghost{}&\pureghost{\Braid^{\left(1\right)}}&\qw&\s{E}\qw&\qw&\ghost{\Braid^{\left(2\right)}}&
			}
		\right\}_{n \in \mathbb{N}}.
	\end{equation*}	
\end{theorem}
\begin{proof}
	The proof of this result consists in going over the first two points of the proof of \autoref{thm:OPT:minimal:transf:goodDeterministic}, and applying them to the case considered here.
\end{proof}

\section{Minimal classical theory}\label{app:mct}
We will now discuss in detail minimal classical theory (MCT). The procedure for the construction of generic OPTs presented in Ref.~\cite{article:BCT} guarantees that the postulates here presented are sufficient to to construct a well-defined operational theory.

\begin{postulate}[Classicality, convexity and type of systems]
	\label{pos:MCT:1}
	The theory \MCT{} is classical, convex and satisfies local discriminability. In addition to the trivial system, for every integer $\sysDimensionD \geq 1$, $\Sys{\MCT{}}$ contains a type of system of size $\sysDimensionD$.
\end{postulate}

\begin{postulate}[Preparation and observation tests]
	\label{pos:MCT:2}
	Given any system $\system{A} \in \Sys{\MCT{}}$, a collection $\eventPreparationTestSystem{\rho}{x}{X}{A} \subset \St{A}$ is a preparation test if and only if $\sum_{x \in X} \rbraketSystem{\observationUniqueDeterministic}{\event{\rho}{x}}{A} = 1$. The observation tests of every system $\system{A} \in \Sys{\MCT{}}$ are all the collections $\eventObservationTestSystem{a}{y}{Y}{A} \subset \EffR{A}$ of generalized effects such that $\left\{ \rbraSystem{\observationEvent{a}{y}}{A} \boxtimes \identityTest{E} \right\}_{\outcomeSingle{y}{Y}} \subset \EffR{AE}$ maps preparation tests of $\system{AE}$ to preparation tests of $\system{E}$ for all $\system{E} \in \Sys{\MCT{}}$. 
\end{postulate}

Where $\EffR{A}$ for every $\system{A} \in \Sys{\MCT{}}$ is defined  by \autoref{pos:MCT:1} through the property of joint perfect discriminability.

\begin{postulate}[Transformations and tests]
	\label{pos:MCT:3}
	The only allowed tests are the ones given by the composition of the elements of
	\begin{equation}	
		\label{eqt:MCT:allowedInstr}
		\left\{ 
		\myQcircuitSupMat{
			&\s{A}\qw&\qw&
		}
		\quad , \quad
		\myQcircuitSupMat{
			&\s{A}\qw&\braidingSym&\s{B}\qw&\qw&
			\\
			&\s{B}\qw&\braidingGhost&\s{A}\qw&\qw&
		}
		\quad , \quad
		\myQcircuitSupMat{
			&\prepareC{\eventTest{\rho}{i}{I}}&\s{A}\qw&\qw&
		}
		\quad , \quad
		\myQcircuitSupMat{
			&\s{A}\qw&\measureD{\observationEventTest{a}{j}{J}}&
		}
		\right\}
	\end{equation} 
	where $\eventTest{\rho}{i}{I}$ and $\observationEventTest{a}{j}{J}$ are all the possible preparation tests and observation tests allowed in the theory by \autoref{pos:MCT:2}, and the limits of all Cauchy sequences of tests of this type. Thus the only allowed transformations are the ones obtainable by sequential and parallel composition of the elements of
	\begin{equation}
		\label{eqt:MCT:allowedTransf}
		\left\{ 
		\myQcircuitSupMat{
			&\s{A}\qw&\qw&
		}
		\quad , \quad
		\myQcircuitSupMat{
			&\s{A}\qw&\braidingSym&\s{B}\qw&\qw&
			\\
			&\s{B}\qw&\braidingGhost&\s{A}\qw&\qw&
		}
		\quad , \quad
		\myQcircuitSupMat{
			&\prepareC{\rho}&\s{A}\qw&\qw&
		}
		\quad , \quad
		\myQcircuitSupMat{
			&\s{A}\qw&\measureD{a}&
		}
		\right\}
	\end{equation}
	for every $\system{A}, \system{B} \in \Sys{\MCT{}}$, $\rket{\eventNoDown{\rho}} \in \StOPT{\MCT{}}$, and $\rbra{a} \in \EffOPT{\MCT{}}$, and the limits of all the Cauchy sequences of events of this type that belong to a test of the theory.
\end{postulate}

We recall the following definitions. 
\begin{definition}[Convex OPT]
	An OPT \OPT{} is \textdef{convex} if $\St{A}$ coincides with its convex hull for all $\system{A} \in \Sys{\OPT}$.
\end{definition}

\begin{definition}[Local discriminability]
	It is possible to discriminate any pair of states of composite systems using only local measurements~\cite{book:quantumTheoryFromFirstPrinciples}.
\end{definition}
On top of this, we remind that this theory is causal since every state is proportional to a deterministic one~\cite{book:quantumTheoryFromFirstPrinciples}. %


\subsection{MCT has full compatibility of observation tests}\label{app:mct:compatibility}
We will now prove that MCT satisfies the property of full compatibility of observation tests. Let us consider two of them. The most generic form they can take in an $n$-dimensional system is the following:
\begin{align*}
	& \eventObservationTestSystem{\observationEventNoDown{a}}{x}{X}{A} = \left\{ \probabilityEventNoDown{p}^{0}_{0} \rbra{\observationEventNoDown{0}} +  \probabilityEventNoDown{p}^{0}_{1} \rbra{\observationEventNoDown{1}} + \cdots +  \probabilityEventNoDown{p}^{0}_{n} \rbra{\observationEventNoDown{n}}, \probabilityEventNoDown{p}^{1}_{0} \rbra{\observationEventNoDown{0}} +  \probabilityEventNoDown{p}^{1}_{1} \rbra{\observationEventNoDown{1}} + \cdots +  \probabilityEventNoDown{p}^{1}_{n} \rbra{n}, \ldots, \probabilityEventNoDown{p}^{m}_{0} \rbra{\observationEventNoDown{0}} +  \probabilityEventNoDown{p}^{m}_{1} \rbra{\observationEventNoDown{1}} + \cdots +  \probabilityEventNoDown{p}^{m}_{n} \rbra{\observationEventNoDown{n}} \right\} \in \Eff{A}\\
	& \eventObservationTestSystem{b}{y}{Y}{A} = \left\{ \probabilityEventNoDown{q}^{0}_{0} \rbra{\observationEventNoDown{0}} +  \probabilityEventNoDown{q}^{0}_{1} \rbra{\observationEventNoDown{1}} + \cdots +  \probabilityEventNoDown{q}^{0}_{n} \rbra{\observationEventNoDown{n}}, \probabilityEventNoDown{q}^{1}_{0} \rbra{\observationEventNoDown{0}} +  \probabilityEventNoDown{q}^{1}_{1} \rbra{\observationEventNoDown{1}} + \cdots +  \probabilityEventNoDown{q}^{1}_{n} \rbra{n}, \ldots, \probabilityEventNoDown{q}^{k}_{0} \rbra{\observationEventNoDown{0}} +  \probabilityEventNoDown{q}^{k}_{1} \rbra{\observationEventNoDown{1}} + \cdots +  \probabilityEventNoDown{p}^{k}_{n} \rbra{\observationEventNoDown{n}} \right\} \in \Eff{A},
\end{align*}
where $\sum_{i = 0}^{m} \probabilityEventNoDown{p}^{i}_{j} = 1 \quad \forall j = 0,\ldots,n$ with $\probabilityEventNoDown{p}^{i}_{j} \in [0,1] \, \forall i,j$ and analogously for $\probabilityEventNoDown{q}^{i}_{j}$. Defining now
\begin{equation*}
	\eventObservationTestSystem{c}{\outcomeDouble{i}{j}}{\outcomeSpaceDouble{I}{J}}{A} = \left\{ \probabilityEventNoDown{r}^{0}_{0} \rbra{\observationEventNoDown{0}}, \probabilityEventNoDown{s}^{0}_{0} \rbra{\observationEventNoDown{0}}, \probabilityEventNoDown{r}^{1}_{0} \rbra{\observationEventNoDown{0}}, \probabilityEventNoDown{s}^{1}_{0} \rbra{\observationEventNoDown{0}}, \ldots, \probabilityEventNoDown{r}^{i}_{j} \rbra{\observationEventNoDown{j}}, \probabilityEventNoDown{s}^{i}_{j} \rbra{\observationEventNoDown{j}}, \ldots  \right\} \in \Eff{A},
\end{equation*}
where $\probabilityEventNoDown{r}^{i}_{j} = \min\left\{ \probabilityEventNoDown{p}^{i}_{j} , \probabilityEventNoDown{q}^{}_{j} \right\}$ and $\probabilityEventNoDown{s}^{i}_{j} = \max\left\{ \probabilityEventNoDown{p}^{i}_{j} , \probabilityEventNoDown{q}^{i}_{j} \right\} - \min\left\{ \probabilityEventNoDown{p}^{i}_{j} , \probabilityEventNoDown{q}^{i}_{j} \right\}$, one can verify from direct calculation that this is an effect of the theory, since it complies with \autoref{pos:MCT:2}, and that it satisfies the following relations
\begin{equation*}
	\label{eqt:mct:compatibleObsInstr}
	\begin{aligned}
		& \rbraSystem{\observationEvent{\observationEventNoDown{a}}{x}}{A} = \sum_{\outcomeDouble{i}{j} \in \outcomeSpaceConditioned{V}{x}} \rbraSystem{\observationEvent{c}{\outcomeDouble{i}{j}}}{A} \quad \forall\outcomeSingle{x}{X}, \\
		&\rbraSystem{\observationEvent{b}{y}}{A} = \sum_{\outcomeDouble{i}{j} \in \outcomeSpaceConditioned{W}{y}} \rbraSystem{\observationEvent{c}{\outcomeDouble{i}{j}}}{A} \quad \forall\outcomeSingle{y}{Y},
	\end{aligned}
\end{equation*}
where the ensembles $\left\{\outcomeSpaceConditioned{V}{x} \right\}_{\outcomeSingle{x}{X}}$ and $\left\{\outcomeSpaceConditioned{W}{y} \right\}_{\outcomeSingle{y}{Y}}$ are appropriate disjoint partitions of $\outcomeSpaceDouble{I}{J}$. The proof is now concluded since we have shown that the two observation tests are compatible.

\subsection{A property of MCT's tests}

An interesting aspect of MCT's tests is that the ancillary system \system{C} can always neglected and considered only through a coarse-graining operation.

To show this it is sufficient to observe that any state of the theory can be uniquely decomposed on the vertices of the simplex which is the state space, and that any effect can be written as coarse graining of the effects that perfectly discriminate the vertices of the simplex:
\begin{align*}
	\rketSystem{\preparationEventNoDown{\rho}}{A} & = \sum_{i = 1}^{\sysDimension{A}} p_{i} \rketSystem{\preparationEventNoDown{i}}{A}, \\
	\rbraSystem{\observationEventNoDown{a}}{A} & = \sum_{j = 1}^{\sysDimension{A}} d_{j} \rbraSystem{\observationEventNoDown{j}}{A},
\end{align*}
where $p_{i} \in [0,1] \, \forall i = 1,\ldots,\sysDimension{A}$ and $\sum_{i = 1}^{\sysDimension{A}} \leq 1$, and $d_{j} \in [0,1] \, \forall j = 1,\ldots,\sysDimension{A}$.\\
This property still holds also for states and effects of composite systems. Therefore, any transformation of MCT obtained as parallel and sequential composition of the elements in \eqref{eqt:MCT:allowedTransf} can be rewritten as
\allowDisplayBreaks{
	\begin{align*}
		& \genericT{\prim{A}}{\prim{B}}{E}{C}{\preparationEventNoDown{\rho}}{\observationEventNoDown{a}}{\Braid^{\left(1\right)}}{\Braid^{\left(2\right)}}  \\[10pt]
		& = \sum_{i = 1}^{\sysDimension{\prim{B}}} \sum_{j = 1}^{\sysDimension{\prim{A}}} \sum_{m = 1}^{\sysDimension{C}} \sum_{\prim{m} = 1}^{\sysDimension{C}} p_{mi} d_{\prim{m}j} \genericT{\prim{A}}{\prim{B}}{E}{C}{\preparationEventNoDown{(mi)}}{(\prim{\observationEventNoDown{m}}\observationEventNoDown{i})}{\Braid^{\left(1\right)}}{\Braid^{\left(2\right)}} \\[10pt]
		& = \sum_{i = 1}^{\sysDimension{\prim{B}}} \sum_{j = 1}^{\sysDimension{\prim{A}}} \sum_{\prim{m} = 1}^{\sysDimension{C}} p_{mi} d_{\prim{m}j} 
		\myQcircuit{
			&\pureghost{}&\prepareC{\preparationEventNoDown{m}}&\qw&\qw&\s{C}\qw&\qw&\qw&\measureD{\prim{\observationEventNoDown{m}}}&
			\\
			&\s{A}\qw&\multigate{1}{\Braid^{\left(1\right)}}&\s{\prim{A}}\qw&\measureD{\observationEventNoDown{j}}&\pureghost{}&\prepareC{\preparationEventNoDown{i}}&\s{\prim{B}}\qw&\multigate{1}{\Braid^{\left(2\right)}}&\s{B}\qw&\qw&
			\\
			&\pureghost{}&\pureghost{\Braid^{\left(1\right)}}&\qw&\qw&\s{E}\qw&\qw&\qw&\ghost{\Braid^{\left(2\right)}}&	
		}  \\[10pt]
		& = \sum_{i = 1}^{\sysDimension{\prim{B}}} \sum_{j = 1}^{\sysDimension{\prim{A}}} \sum_{m = 1}^{\sysDimension{C}} \sum_{\prim{m} = 1}^{\sysDimension{C}} p_{mi} d_{\prim{m}j} \kronekerDelta{m}{\prim{m}} 
		\myQcircuit{
			&\s{A}\qw&\multigate{1}{\Braid^{\left(1\right)}}&\s{\prim{A}}\qw&\measureD{\observationEventNoDown{j}}&\pureghost{}&\prepareC{\preparationEventNoDown{i}}&\s{\prim{B}}\qw&\multigate{1}{\Braid^{\left(2\right)}}&\s{B}\qw&\qw&
			\\	&\pureghost{}&\pureghost{\Braid^{\left(1\right)}}&\qw&\qw&\s{E}\qw&\qw&\qw&\ghost{\Braid^{\left(2\right)}}&	
		}  \\[10pt]
		& = \sum_{i = 1}^{\sysDimension{\prim{B}}} \sum_{j = 1}^{\sysDimension{\prim{A}}} \sum_{m = 1}^{\sysDimension{C}} p_{mi} d_{mj}
		\myQcircuit{
			&\s{A}\qw&\multigate{1}{\Braid^{\left(1\right)}}&\s{\prim{A}}\qw&\measureD{\observationEventNoDown{j}}&\pureghost{}&\prepareC{\preparationEventNoDown{i}}&\s{\prim{B}}\qw&\multigate{1}{\Braid^{\left(2\right)}}&\s{B}\qw&\qw&
			\\
			&\pureghost{}&\pureghost{\Braid^{\left(1\right)}}&\qw&\qw&\s{E}\qw&\qw&\qw&\ghost{\Braid^{\left(2\right)}}&	
		}.
	\end{align*}
}


\twocolumngrid

\bibliography{biblio}
\bibliographystyle{apsrev4-2.bst}

\end{document}